\documentclass[lettersize,journal]{IEEEtran} 

\usepackage{tikz}
\tikzstyle{arrow} = [thick,->,>=stealth]
\pdfoutput=1
\usepackage{amsmath,amsthm,amssymb,amsfonts,bm,dsfont}
\usepackage[bookmarksopen=true,pagebackref=true]{hyperref}
\usepackage{cleveref}
\usepackage{graphicx}
\usepackage{caption}
\usepackage[labelformat=simple]{subcaption}

\usepackage{enumerate}
\usepackage[noadjust]{cite}
\allowdisplaybreaks[2]
\usepackage{balance}

\newcommand{\cC}{\mathcal{C}}
\newcommand{\cD}{\mathcal{D}}
\newcommand{\cS}{\mathcal{S}}
\newcommand{\cT}{\mathcal{T}}
\newcommand{\cX}{\mathcal{X}}
\newcommand{\cY}{\mathcal{Y}}
\newcommand{\bE}{\mathbb{E}}

\theoremstyle{plain}
\newtheorem{theorem}{Theorem}
\newtheorem{lemma}[theorem]{Lemma}

\theoremstyle{definition}

\theoremstyle{remark}
\newtheorem{remark}{Remark}

\begin{document}
	\title{Semantic Compression with Side Information: A Rate-Distortion Perspective}
	
	\author{Tao Guo, {\it Member, IEEE}, Yizhu Wang, Jie Han, Huihui Wu, {\it Member, IEEE}, \\
		Bo Bai, {\it Senior Member, IEEE}, and Wei Han, {\it Member, IEEE}}
	\date{\today}
	
	\maketitle
	\balance
	
	\begin{abstract}
We consider the semantic rate-distortion problem motivated by task-oriented video compression. The semantic information corresponding to the task, which is not observable to the encoder, shows impacts on the observations through a joint probability distribution. The similarities among intra-frame segments and inter-frames in video compression are formulated as side information available at both the encoder and the decoder. The decoder is interested in recovering the observation and making an inference of the semantic information under certain distortion constraints.

We establish the information-theoretic limits for the tradeoff between compression rates and distortions by fully characterizing the rate-distortion function. We further evaluate the rate-distortion function under specific Markov conditions for three scenarios: i) both the task and the observation are binary sources; ii) the task is a binary classification of an integer observation as even and odd; iii) Gaussian correlated task and observation. We also illustrate through numerical results that recovering only the semantic information can reduce the coding rate comparing to recovering the source observation.
	\end{abstract}
	
	\begin{IEEEkeywords}
		Semantic communication, inference, video compression, rate-distortion function.
	\end{IEEEkeywords}

	\section{Introduction}
	The rate limit for lossless compression of memoryless sources is commonly known as the {\it entropy} shown by Shannon in his landmark paper~\cite{Shannon-1948}. In addition, 
	lossy source coding under  given fidelity criterion was also  introduced in the same paper. Further, the Shannon  {\it rate-distortion function} was  proposed in ~\cite{Shannon-RD-1959}, characterizing the optimal tradeoff between compression rates and distortion measurements, from the perspective of mutual information.  
	
	Thereafter, the rate-distortion function was investigated when side information is available at the encoder or/and decoder, see  \cite{Gray-cRD-1972,Gray-cRD-1973,Wyner-Ziv-IT76,Wyner-RD-SI-78,Kaspi-94,Permuter-Weissman-RDCSI-11,Watanabe-product_source-13,Heegard-Berger-85,complementary_delivery-06} and reference therein. In the case that the side information is only available at the encoder, then no benefit could be achieved. In case of side information being only at the decoder, the corresponding rate-distortion function was considered by Wyner and Ziv in \cite{Wyner-Ziv-IT76}, with its extensions being discussed in \cite{Watanabe-product_source-13,Heegard-Berger-85,complementary_delivery-06,Wyner-RD-SI-78,Kaspi-94,Permuter-Weissman-RDCSI-11}.
	Finally, if both the encoder and  decoder have access to the same side information, the optimal tradeoff is called  {\it conditional rate-distortion function}, which was given by  \cite{Gray-cRD-1972} and \cite{Gray-cRD-1973}. 
	The lossy source coding theory finds applications in establishing information theoretic limits for practical  compression of speech signals, images and  videos etc.~ \cite{Tasto-Wintz-72,Aaron-Rane-Zhang-Girod-03,Gibson-Hu-book-14,Wang-etal-MDforVideo-05}. 
	Practical techniques for video compression have been explored since decades ago \cite{Chen-etal-91,Zeng-02,Habibian-etal-19}. 
	Currently, popular protocols such as HEVC, VP9, VVC and AV1 
	are based on partitioning a picture/frame into coding tree units, which typically correspond to 64x64 or 128x128 pixel areas.
	Each coding tree unit is then partitioned into coding blocks (segments) \cite{Habibi-Wintz-71,Hu-Gibson-09}, following a recursive coding tree representation. The  aforementioned compression schemes consider  both 
	intra-correlation within one frame and inter-correlation between two consecutive frames. 
	
	
	Nowadays, with the development of high-definition videos, 5G communication systems and industrial Internet of Things, communication overhead and storage demand have been  exponentially growing. As a result, higher compression  rates are required, but it seems not possible by simply  compressing a given source (e.g., a video or image) itself in light of the rate-distortion limits. 
	
	Semantic or Task-oriented compression \cite{Yener-etal-semantic-18,RD-Perception-19,Xie-Qin-21,Wen-Qin-21}, aiming at efficiently compressing the sources  according to specific communication tasks (e.g., video detection, inference, classification, decision making, etc.), has been viewed as a promising technique for future 6G systems due to its extraordinary effectiveness. 
	Particularly,  the goal of semantic compression is to recover the necessary semantic information  corresponding to a certain task instead of each individually transmitted bit as in Shannon communication setups, and thus it leads to significant reduction of coding rates.  
	
	In addition to the interested semantic information, the original sources are also required in some cases such as video surveillance, for the purpose of evidence storage and verification. An effective way is to save only the most important or relevant segments of a video,  and related work on semantic-based  image/video segmentation can be found in \cite{Wang-Wang-Chen-21,Guo-etal-SemanticSeg-17,Zhao-etal-SemanticSeg-17}.
	Most recently, the classical indirect source coding problem~\cite{semantic-indirect-80} was revisited from the semantic point of view, and  the corresponding rate-distortion framework for semantic information in \cite{Liu-Zhang-Poor-ISIT21,Kountouris-ISIT22,Liu-Zhang-Poor-arxiv22,Kountouris-arxiv22}. 
	%
	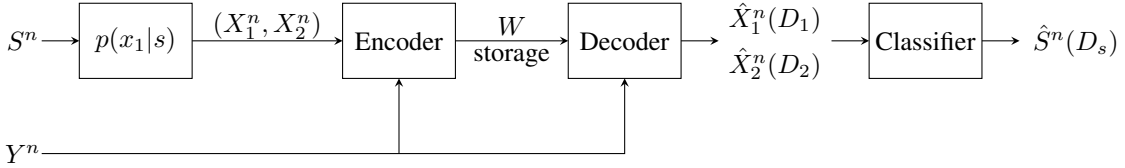
\begin{figure*}[h!]
		\centering
		\begin{tikzpicture}
		\node at (0.75,0) {$S^n$};
		\draw[->,>=stealth] (1,0)--(1.5,0);
		
		\node at (2.25,0) {$p(x_1|s)$};
		\draw (1.5,-0.5) rectangle (3,0.5);
		
		\node at (4.0,0.2) {$(X_1^n,X_2^n)$};
		\draw[->,>=stealth] (3,0)--(5,0);
		
		\node at (5.75,0) {Encoder};
		\draw (5,-0.5) rectangle (6.5,0.5);
		
		\draw[->,>=stealth] (6.5,0)--(8,0);
		\node at (7.25,0.2) {$W$};
		\node at (7.25,-0.2) {storage};
		
		\node at (8.75,0) {Decoder};
		\draw (8.0,-0.5) rectangle (9.5,0.5);
		
		\draw[->,>=stealth] (9.5,0)--(10,0);
		\node at (10.75,0.3) {$\hat{X}_1^n (D_1)$};
		\node at (10.75,-0.3) {$\hat{X}_2^n (D_2)$};
		
		\draw[->,>=stealth] (11.5,0)--(12,0);
		
		\node at (12.75,0) {Classifier};
		\draw (12,-0.5) rectangle (13.5,0.5);
		
		\draw[->,>=stealth] (13.5,0)--(14,0);
		\node at (14.75,0) {$\hat{S}^n (D_s)$};
		
		\node at (0.75,-1.5) {$Y^n$};
		\draw[->,>=stealth] (1.0,-1.5)--(8.75,-1.5)--(8.75,-0.5);
		\draw[->,>=stealth] (5.75,-1.5)--(5.75,-0.5);
		
		\end{tikzpicture}
		\caption{Illustration of system model with side information.}
		\label{fig_seg-inference-SI-model}
	\end{figure*}
	%

	The current paper introduces side information into the  framework of  \cite{Liu-Zhang-Poor-ISIT21} and completely characterizes the semantic rate-distortion problem with side information.  Motivated from the task-oriented video compression,  the semantic information corresponding to the task is not observable to the encoder.  In light of video segmentation,  the observed source is partitioned into two parts, and the semantic information only shows influence  on the more important part. 
	Moreover, the intra-correlation and inter-correlation  are viewed as side information, and they are available at both the encoder and decoder to help compression.  
	The decoder needs to reconstruct the whole source subject to different distortion constraints for the two parts, respectively. 
	The semantic information can be recovered upon observing the source reconstructions at the decoder. Finally,  
	our main contributions are summarized as follows: 
	\begin{enumerate}[1)]
		\item We fully characterize the optimal rate-distortion tradeoff. 
		It is further shown that separately compressing the two source parts is optimal, if the they are independent conditioning on the side information. 
		
		\item The rate-distortion function is evaluated for the inference of a binary source under some specific Markov chains. 
		
		\item We further evaluate the rate-distortion function for binary classification of an integer source. 
		The numerical results show that recovering only the semantic information can reduce the coding rate comparing to recovering the source message. 
		
		\item The rate-distortion function for Gaussian sources is also illustrated, which may provide more insights for future real video compression simulations.
	\end{enumerate}
	This paper is mainly pertained to the information-theoretic aspects, 
	and future work on the limit of real video compression is under investigation.
	
	The rest of the paper is organized as follows. 
	We first formulate the problem and present some preliminary results in \Cref{section:formulation}. 
	In \Cref{section:RD-function}, we characterize the rate-distortion function and some useful properties. 
	Evaluations of the rate-distortion function for binary, integer, and Gaussian sources with Hamming/mean squared error distortions are devoted to \Cref{section:example-binary,section:example-integer,section:example-Gaussian}, respectively. 
	We present and analyze some plots of the evaluations in \Cref{section:plots}. 
	The paper is concluded in \Cref{section:conclusion}. 
	Some essential proofs can be found in the appendices.

	\section{Problem Formulation and Preliminaries}\label{section:formulation}
	\subsection{Problem Formulation}\label{section:problemSetup}
	Consider the system model for video detection (inference) that also requires evidence storage depicted in Fig.~\ref{fig_seg-inference-SI-model}. 
	The problem is defined as follows. 
	A collection of {\it discrete memoryless sources} (DMS) is described by generic random variables $(S,X_1,X_2,Y)$ 
	taking values in finite alphabets $\cS\times\cX_1\times\cX_2\times\cY$ according to probability distribution $p(x_1,x_2,y)p(s|x_1)$. 
	In particular, this indicates the Markov chain $S-X_1-(X_2,Y)$. 
	We interpret $S$ as a latent variable, which is not observable by the encoder. 
	It can be viewed as the semantic information (e.g., the state of a system), 
	which describes the features of the system. 
	We assume that the observation of the system consists of two parts: 
	\begin{itemize}
		\item $X_1$ varies according to the semantic information~$S$, which captures the ``appearance" of the features, 
		e.g., the vehicle and red lights in the frame that captures a violation at the cross; 
		
		\item $X_2$ is the background information irrelevant to the features, 
		e.g., buildings in the frame capturing the violation. 
	\end{itemize}
	$Y$ is the side information that can help compressing such as previous frames in the video. 
	For length-$n$ source sequences, $(S^n,X_1^n,X_2^n,Y^n)$, 
	the encoder has access to only the observed ones $(X_1^n,X_2^n,Y^n)$ and encodes them as $W$ which will be stored at the server. 
	Upon observing local information $Y^n$ and receiving $W$, the decoder reconstructs the source sequences as $(\hat{X}_1^n,\hat{X}_2^n)$ drawn values from $\hat{\cX}_1\times \hat{\cX}_2$, within distortions $D_1$ and $D_2$. 
	Given the reconstructions, the classifier is required to recover the semantic information as~$\hat{S}^n$ from alphabet $\hat{\cS}$ with distortion constraint $D_s$. 
	Here, for simplicity, we assume a perfect classifier, i.e., it is equivalent to recover $\hat{S}^n$ directly at the decoder as illustrated in Fig.~\ref{fig_seg-inference-SI-model-eq}. 
	%
	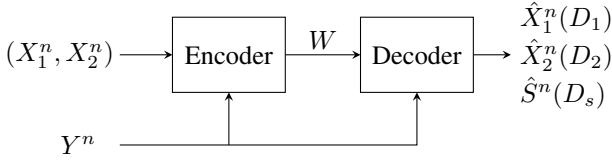
\begin{figure}[t!]
		\centering
		\begin{tikzpicture}
		\node at (0,0) {$(X_1^n,X_2^n)$};
		\draw[->,>=stealth] (0.8,0)--(1.5,0);
		
		\node at (2.25,0) {Encoder};
		\draw (1.5,-0.5) rectangle (3.0,0.5);
		
		\draw[->,>=stealth] (3.0,0)--(4,0);
		\node at (3.5,0.2) {$W$};
		
		\node at (4.75,0) {Decoder};
		\draw (4.0,-0.5) rectangle (5.5,0.5);
		
		\draw[->,>=stealth] (5.5,0)--(6.0,0);
		\node [right] at (6.0,0.5) {$\hat{X}_1^n (D_1)$};
		\node [right] at (6.0,0) {$\hat{X}_2^n (D_2)$};
		\node [right] at (6.0,-0.5) {$\hat{S}^n (D_s)$};
		
		\node at (0.25,-1.2) {$Y^n$};
		\draw[->,>=stealth] (0.8,-1.2)--(4.75,-1.2)--(4.75,-0.5);
		\draw[->,>=stealth] (2.25,-1.2)--(2.25,-0.5);
		
		\end{tikzpicture}
		\caption{The equivalent system model.}
		\label{fig_seg-inference-SI-model-eq}
	\end{figure}
	%
	
	Formally, an $\left(n,2^{nR}\right)$ code is defined by the encoding function 
	\begin{align*}
	En: \cX_1^n\times \cX_2^n \times\cY^n \rightarrow \{1,2,\cdots,2^{nR}\}
	\end{align*}
	and the decoding function
	\begin{align*}
	De: \{1,2,\cdots,2^{nR}\}\times\cY^n \rightarrow \hat{\cX}_1^n\times\hat{\cX}_2^n\times\hat{\cS}^n. 
	\end{align*}
	Let $\mathbb{R}^+$ be the set of nonnegative real numbers. 
	We consider bounded per-letter distortion functions $d_1:\cX_1\times\hat{\cX}_1\rightarrow\mathbb{R}^+$, $d_2:\cX_2\times\hat{\cX}_2\rightarrow\mathbb{R}^+$, and $d_s:\cS\times\hat{\cS}\rightarrow\mathbb{R}^+$. 
	The distortions between length-$n$ sequences are defined by 
	\begin{align*}
	d_1(x_1^n,\hat{x}_1^n) &\triangleq \frac{1}{n}\sum_{i=1}^n d_1(x_{1,i},\hat{x}_{1,i}),   \\
	d_2(x_2^n,\hat{x}_2^n) &\triangleq \frac{1}{n}\sum_{i=1}^n d_2(x_{2,i},\hat{x}_{2,i}),    \\
	d_s(s^n,\hat{s}^n) &\triangleq \frac{1}{n}\sum_{i=1}^n d_s(s_i,\hat{s}_i). 
	\end{align*}
	A nonnegative rate-distortion tuple $(R,D_1,D_2,D_s)$ is said to be {\it achievable} if for sufficiently large $n$, 
	there exists an $\left(n,2^{nR}\right)$ code such that 
	\begin{align*}
	\lim_{n\rightarrow\infty} \bE d_1(X_1^n,\hat{X}_1^n)&\leq D_1,  \\
	\lim_{n\rightarrow\infty} \bE d_2(X_2^n,\hat{X}_2^n)&\leq D_2,  \\
	\lim_{n\rightarrow\infty} \bE d_s(S^n,\hat{S}^n)&\leq D_s. 
	\end{align*}
	The rate-distortion function $R(D_1,D_2,D_s)$ is the infimum of coding rate $R$ for distortions $(D_1,D_2,D_s)$ 
	such that the rate-distortion tuple $(R,D_1,D_2,D_s)$ is achievable. 
	Our goal is to characterize the rate-distortion function. 

	\subsection{Preliminaries}
	\subsubsection{Conditional rate-distortion function}
	The elegant rate-distortion function was investigated and fully characterized in~\cite{Shannon-RD-1959}. 
	Assume the length-$n$ source sequence $X^n$ is independent and identically distributed (i.i.d.) over $\cX$ with generic random variable $X$ 
	and $d:\cX\times\hat{\cX}\rightarrow\mathbb{R}^+$ be a bounded per-letter distortion measure. 
	The rate-distortion function for a given distortion criterion $D$ is given by 
	\begin{equation}
	R(D)=\min_{p(\hat{x}|x):\bE d(X,\hat{X})\leq D} I(X;\hat{X}).  \label{basic-RDfn}
	\end{equation}
	It was proved in \cite{Shannon-RD-1959} and also introduced in \cite{Berger-book-71,Berger-book-78,Raymond-book,NetworkIT-book,Thomas-Cover-book} that $R(D)$ is a {\it non-increasing} and {\it convex} function of $D$. 
	
	If both the encoder and decoder are allowed to observe side information $Y^n$ (with generic variable $Y$ over $\cY$ jointly distributed with $X$), as depicted in Fig.~\ref{fig_RD-SI-model}, 
	then the tradeoff is called the conditional rate-distortion function~\cite{Berger-book-71,Gray-cRD-1972,Gray-cRD-1973}, which is characterized as 
	\begin{equation}
	R_{X|Y}(D)=\min_{p(\hat{x}|x,y):\bE d(X,\hat{X})\leq D} I(X;\hat{X}|Y).  \label{basic-RDfn-SI}
	\end{equation}
	%
	\begin{figure}[t!]
		\centering
		\begin{tikzpicture}
		\node at (3.5,0) {$X^n$};
		\draw[->,>=stealth] (4.0,0)--(5,0);
		
		\node at (5.75,0) {Encoder};
		\draw (5,-0.5) rectangle (6.5,0.5);
		
		\draw[->,>=stealth] (6.5,0)--(8,0);
		\node at (7.25,0.2) {$W$};
		
		\node at (8.75,0) {Decoder};
		\draw (8.0,-0.5) rectangle (9.5,0.5);
		
		\draw[->,>=stealth] (9.5,0)--(10,0);
		\node at (10.75,0) {$\hat{X}^n (D)$};
		
		\node at (3.5,-1.25) {$Y^n$};
		\draw[->,>=stealth] (4.0,-1.25)--(8.75,-1.25)--(8.75,-0.5);
		\draw[->,>=stealth] (5.75,-1.25)--(5.75,-0.5);
		
		\end{tikzpicture}
		\caption{Conditional Rate-distortion model.}
		\label{fig_RD-SI-model}
	\end{figure}
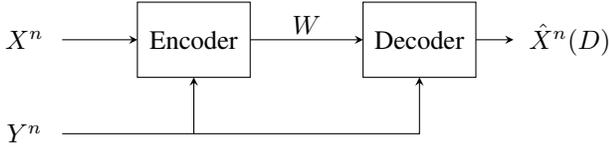
	%
	It is shown in \cite{Gray-cRD-1972} that the conditional rate-distortion function can also be obtained as 
	the weighted sum of the marginal rate-distortion function of sources with distribution $P_{X|Y}(\cdot|y),y\in\cY$, i.e., 
	\begin{equation}
	R_{X|Y}(D) = \min_{\{D_y:y\in\cY\}:\sum_{y\in\cY}p(y)\cdot D_y\leq D} \sum_{y\in\cY}p(y)\cdot R(D_y),  \label{basic-RDfn-SI-alt}
	\end{equation}
	where for any $y\in\cY$, $R(D_y)$ is obtained from \eqref{basic-RDfn} through replacing the source distribution by $P_{X|Y}(\cdot|y)$. 
	This property will be useful for evaluating conditional rate-distortion functions of given source distributions. 
	If $(X,Y)$ is a {\it doubly symmetric binary source} (DSBS) with parameter $p_0$, i.e., 
	\begin{equation}
	p(x,y)= \left[\begin{array}{cc}
	\frac{1-p_0}{2} & \frac{p_0}{2} \\
	\frac{p_0}{2} & \frac{1-p_0}{2}
	\end{array}\right],  \label{DSB-ditrbution}
	\end{equation}
	then the conditional rate distortion function is given in \cite{Gray-cRD-1973} by 
	\begin{equation}
	R_{X|Y}(D) = \left[h_b(p_0) - h_b(D)\right] \cdot \mathds{1}_{_{0\leq D\leq p_0}},   \label{basic-RDfn-SI-binary}
	\end{equation}
	where $h_b(q)=-q\log q - (1-q)\log(1-q)$ is the entropy for a Bernoulli($q$) distribution 
	and $\mathds{1}_{A}$ is the indicator function of whether event $A$ happens. 
	
	\subsubsection{Rate-distortion function with two constraints}
	The scenario was discussed in \cite{ElGamal-Cover-82,Thomas-Cover-book} that we wish to describe the i.i.d. source sequence $X^n$ at rate $R$
	and recover two reconstructions $\hat{X}_a^n$ and $\hat{X}_b^n$ with distortion criteria $\bE d_a(X^n,\hat{X}_a^n)\leq D_a$ and $\bE d_b(X^n,\hat{X}_b^n)\leq D_b$, respectively. 
	The rate-distortion function is given by 
	\begin{equation}
	R_{\text{2d}}(D_a,D_b)= \min_{\substack{p(\hat{x}_a,\hat{x}_b|x):~\\ \bE d_a(X,\hat{X}_a)\leq D_a \\ \bE d_b(X,\hat{X}_b)\leq D_b}} I(X;\hat{X}_a,\hat{X}_b).  \label{basic-RDfn-2d}
	\end{equation}
	Comparing \eqref{basic-RDfn} and \eqref{basic-RDfn-2d}, we easily see that 
	\begin{equation*}
	\max\{R(D_a), R(D_b)\}\leq R_{\text{2d}}(D_a,D_b) \leq R(D_a)+R(D_b). 
	\end{equation*}
	For the special case where $\hat{\cX}_a=\hat{\cX}_b$ and $d_a(x,\hat{x})=d_b(x,\hat{x})$ for all $x\in\cX$ and $\hat{x}\in\hat{\cX}_a$, 
	it suffice to recover only one sequence $\hat{X}_a^n=\hat{X}_b^n$ with distortion $\min\{D_a,D_b\}$. 
	Then both distortion constraints are satisfied since 
	\begin{equation*}
	\bE d_a(X^n,\hat{X}_a^n) = \min\{D_a,D_b\} \leq D_a 
	\end{equation*}
	and
	\begin{equation*}
	\bE d_b(X^n,\hat{X}_b^n) = \min\{D_a,D_b\} \leq D_b. 
	\end{equation*}
	This implies 
	\begin{equation}
	R_{\text{2d}}(D_a,D_b)= R(\min\{D_a,D_b\}) = \max\{R(D_a),R(D_b)\},  \label{basic-RDfn-2d-minD}
	\end{equation}
	where the second equality follows from the non-increasing property of $R(D)$. 
	
	When side information is available at the decoder for only one of the two reconstructions, e.g., $\hat{X}_b$, 
	it was proved in \cite{Heegard-Berger-85,Kaspi-94} that successive encoding (first $\hat{X}_a$, then $\hat{X}_b$) is optimal. 
	For the case when the two reconstructions have access to different side information respectively, 
	the rate-distortion tradeoff was characterized in \cite{Heegard-Berger-85,complementary_delivery-06,Watanabe-product_source-13,NetworkIT-book}.

	\subsubsection{Rate-distortion function of two sources}
	The problem of compressing two i.i.d. source sequences $X_a^n$ and $X_b^n$ at the same encoder is considered in \cite[Problem. 10.14]{Thomas-Cover-book}. 
	The rate-distortion function is given therein, which is 
	\begin{equation}
	R_{\text{2s}}(D_a,D_b)= \min_{\substack{p(\hat{x}_a,\hat{x}_b|x_a,x_b):~\\ \bE d_a(X_a,\hat{X}_a)\leq D_a \\ \bE d_b(X_b,\hat{X}_b)\leq D_b}} I(X_a,X_b;\hat{X}_a,\hat{X}_b).  \label{basic-RDfn-2s}
	\end{equation}
	It is also shown that for two independent sources, compressing simultaneously is the same as compressing separately in terms of the rate and distortions, i.e., 
	\begin{equation}
	R_{\text{2s}}(D_a,D_b)= R(D_a) + R(D_b). \label{RD-2sources-indep}
	\end{equation}
	If the two sources are dependent, the equality in \eqref{RD-2sources-indep} can be false, and the Slepian-Wolf rate region~\cite{Slepian-Wolf-73} indicates that joint entropy of the two source variables is sufficient and optimal for lossless reconstructions. 
	Taking into account distortions, Gray showed via an example in~\cite{Gray-cRD-1973} that the compression rate can be strictly larger than $R(D_a) + R_{X_b|X_a}(D_b)$ in general. 
	At last, some related results for compressing compound sources can be found in~\cite{composite-sources-84}.

	\section{Optimal Rate-distortion Tradeoff}\label{section:RD-function}
	\subsection{The Rate-distortion Function}
	
	\begin{theorem}\label{thm:main}
		The rate-distortion function for compression and inference with side information is given as the solution to the following optimization problem
		\begin{align}
		R(D_1,D_2,D_s) &= \min I(X_1,X_2;\hat{X}_1,\hat{X}_2,\hat{S}|Y)  \label{RD-fn_rate} \\
		& \quad \textup{ s.t. } \bE d_1(X_1,\hat{X}_1)\leq D_1  \label{RD-fn_constraint-1} \\
		& \qquad\quad \bE d_2(X_2,\hat{X}_2)\leq D_2  \label{RD-fn_constraint-2} \\
		& \qquad\quad \bE d_s'(X_1,\hat{S})\leq D_s, \label{RD-fn_constraint-3}
		\end{align} 
		where the minimum is taken over all conditional pmf $p(\hat{x}_1,\hat{x}_2,\hat{s}|x_1,x_2,y)$ and 
		the modified distortion measure is defined by 
		\begin{equation}
		d_s'(x_1,\hat{s})=\frac{1}{p(x_1)}\sum_{s\in\cS}p(x_1,s)d_s(s,\hat{s}).   \label{distortion-s-alt}
		\end{equation}
	\end{theorem}
	\begin{proof}
		We can interpret the problem as the combination of rate-distortion with two sources ($X_1$ and $X_2$), rate-distortion with two constraints ($X_1$ is recovered with two constraints $D_1$ and $D_s$), and conditional rate-distortion (conditioning on $Y$). 
		Then the theorem can be obtained informally by combining the rate-distortion functions in \eqref{basic-RDfn-SI}, \eqref{basic-RDfn-2d}, and \eqref{basic-RDfn-2s}.
		For completeness, we provide a rigorous technical proof in Appendix~\ref{proof-thm:main}. 
	\end{proof}

	\subsection{Some Properties}
	Similar to the rate-distortion function in \eqref{basic-RDfn}, we collect some properties in the following lemma. 
	The proof simply follows the same procedure as that for \eqref{basic-RDfn} in \cite{Shannon-RD-1959,Berger-book-71,Berger-book-78,Raymond-book,Thomas-Cover-book,NetworkIT-book}. 
	We omit the details here. 
	
	\begin{lemma}\label{lemma:RD-property}
		The rate-distortion function $R(D_1,D_2,D_s)$ is non-increasing and convex in $(D_1,D_2,D_s)$. 
	\end{lemma}
	
	Recall from \eqref{RD-2sources-indep} that compressing two independent sources is the same as compressing them simultaneously. 
	Then one may query that whether the optimality of separate compression remains to hold here? 
	We answer the question in the following lemma. 
	\begin{lemma}\label{lemma:opt-separate-compression}
		If $X_1-Y-X_2$ forms a Markov chain, then 
		\begin{equation*}
		R(D_1,D_2,D_s)= R_{\text{2d},X_1|Y}(D_1,D_s) + R_{X_2|Y}(D_2),
		\end{equation*}
		where the conditional rate-distortion function with two constraints is given by 
		\begin{equation*}
		R_{\text{2d},X_1|Y}(D_1,D_s) = \min_{\substack{p(\hat{x}_1,\hat{s}|x_1,y): \\ \bE d_1(X_1,\hat{X}_1)\leq D_1 \\ \bE d_s'(X_1,\hat{S})\leq D_s}} I(X_1;\hat{X}_1,\hat{S}|Y) 
		\end{equation*}
		and the conditional rate-distortion function is given in \eqref{basic-RDfn-SI} and can be written by 
		\begin{equation*}
		R_{X_2|Y}(D)=\min_{p(\hat{x}_2|x_2,y):\bE d_2(X_2,\hat{X}_2)\leq D_2} I(X_2;\hat{X}_2|Y). 
		\end{equation*}
	\end{lemma}
	\begin{proof}
		The proof is given in Appendix~\ref{proof-lemma:opt-separate-compression}. 
	\end{proof}
	
	\begin{remark}
		Compared to the independence assumption for the equality in \eqref{RD-2sources-indep}, we have the {\it conditional} independence $I(X_1;X_2|Y)=0$ in \Cref{lemma:opt-separate-compression}. 
		This is intuitive since compared to the setting for \eqref{RD-2sources-indep}, we have the additional side information available to the encoder and decoder. 
	\end{remark}
	
	\subsection{Rate-distortion Function for Semantic Information}
	The indirect rate-distortion problem can be viewed as a special case of \Cref{thm:main} that only recovers the semantic information $S$, 
	i.e., $X_2$ and $Y$ are constants and $D_1=\infty$. 
	Denote the minimum achievable rate for a given distortion constraint $D_s$ by $R_s(D_s)$. 
	
	Consider the binary sources and assume $S$ and $X_1$ follow the doubly symmetric binary distribution, i.e., 
	\begin{equation}
	p(s,x_1)=\left[
	\begin{matrix}
	\frac{1-p}{2} & \frac{p}{2} \\
	\frac{p}{2} & \frac{1-p}{2} 
	\end{matrix}
	\right].  \label{DSB-ditrbution-SX1}
	\end{equation}
	The transition probability $p(x_1|s)$ can also be defined via the binary symmetric channel (BSC) in Fig.~\ref{fig_BSC}.
	\begin{figure}[!h]
		\centering
		\begin{tikzpicture}
		\node at (0,0) {$S$};
		\node at (0.3,0.5) {$0$};
		\node at (0.3,-0.5) {$1$};
		
		\draw (0.5,-0.5)--(2.5,-0.5)--(0.5,0.5)--(2.5,0.5)--cycle;
		
		\node at (3.0,0) {$X_1$};
		\node at (2.7,0.5) {$0$};
		\node at (2.7,-0.5) {$1$};
		
		\node at (1.5,0.75) {$1-p$};
		\node at (1.5,-0.75) {$1-p$};
		\node at (1.0,0) {$p$};
		\end{tikzpicture}
		\caption{Transition probability from $S$ to $X_1$: in terms of BSC.}
		\label{fig_BSC}
	\end{figure}
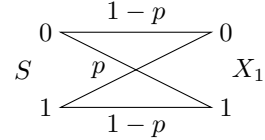
	Assume $p\leq 0.5$, which means that $X_1$ has a higher probability to reflect the same value as $S$. 
	Let $d_s:\cS\times\hat{\cS}\rightarrow\{0,1\}$ be the Hamming distortion measure. 
	Then the evaluation of $R_s(D_s)$ is given in the following lemma. 
	
	Let $R_{d_s'}(\cdot)$ be the ordinary rate-distortion function in \eqref{basic-RDfn} under the distortion measure $d_s'$ (c.f. \eqref{distortion-s-alt}).
	For notational simplicity, for $D_s\geq p$, define 
	\begin{equation}
	D_s^0\triangleq \frac{D_s-p}{1-2p}.  \label{def-Ds0}
	\end{equation}
	
	\begin{lemma}\label{lemma:RD-semantic}
		For binary sources in \eqref{DSB-ditrbution-SX1} and Hamming distortion, the rate-distortion function for semantic information is 
		\begin{equation}
		R_s(D_s)=R_{d_s'}(D_s),   \label{RD-fn-semantic}
		\end{equation}
		where $R_{d_s'}(D_s)=R\left(D_s^0\right)=\left[1-h_b\left(\frac{D_s-p}{1-2p}\right)\right]\cdot \mathds{1}_{_{p\leq D_s\leq 0.5}}.$
	\end{lemma}
	\begin{proof}
		The evaluation of $R_s(D_s)$ was given in \cite{Berger-book-71,Kipnis-ITW15}. 
		A simpler proof can be found in Appendix~\ref{proof-lemma:RD-semantic}. 
	\end{proof}
	\begin{remark}
		By the properties of the rate-distortion function in~\eqref{basic-RDfn} and the linearity between $D_s^0$ and $D_s$, we see that $R_s(D_s)$ is also non-increasing and convex in $D_s$. 
	\end{remark}
	
	\begin{remark}
		It is easy to check that $\frac{D-p}{1-2p}< D$ for $D<0.5$. 
		This implies that $R_s(D)>R(D) $ for $D<0.5$, where $R(D)$ is the ordinary rate-distortion function in \eqref{basic-RDfn}. 
		The inequality is intuitive from the data processing inequality that under the same distortion constraint $D$, recovering $S$ directly (with rate $R(D)$) is easier than recovering it from the observation $X_1$ (with rate $R_s(D)$). 
		Moreover, we see from the lemma that $D_s\geq p$, which means that the semantic information can never be losslessly recovered for $p>0$. 
		This can be induced from the fact that even we know the complete information of $X_1$, the best distortion for reconstructing $S$ is the distortion between $S$ and $X_1$ which is equal to $p$. 
		The rate-distortion functions $R_s(D)$ and $R(D)$ are illustrated in Fig.~\ref{fig_compare-RsD-RD} for $p=0.1$, which verifies the above observations. 
		For general source and distortion measure, we have $R_s(D)\geq R(D)$ where the equality holds only when $X_1$ determines $S$. 
		This can be easily proved by the data processing inequality and we omit the details here. 
		\begin{figure}[!t]
			\centering
			\includegraphics[scale=0.5]{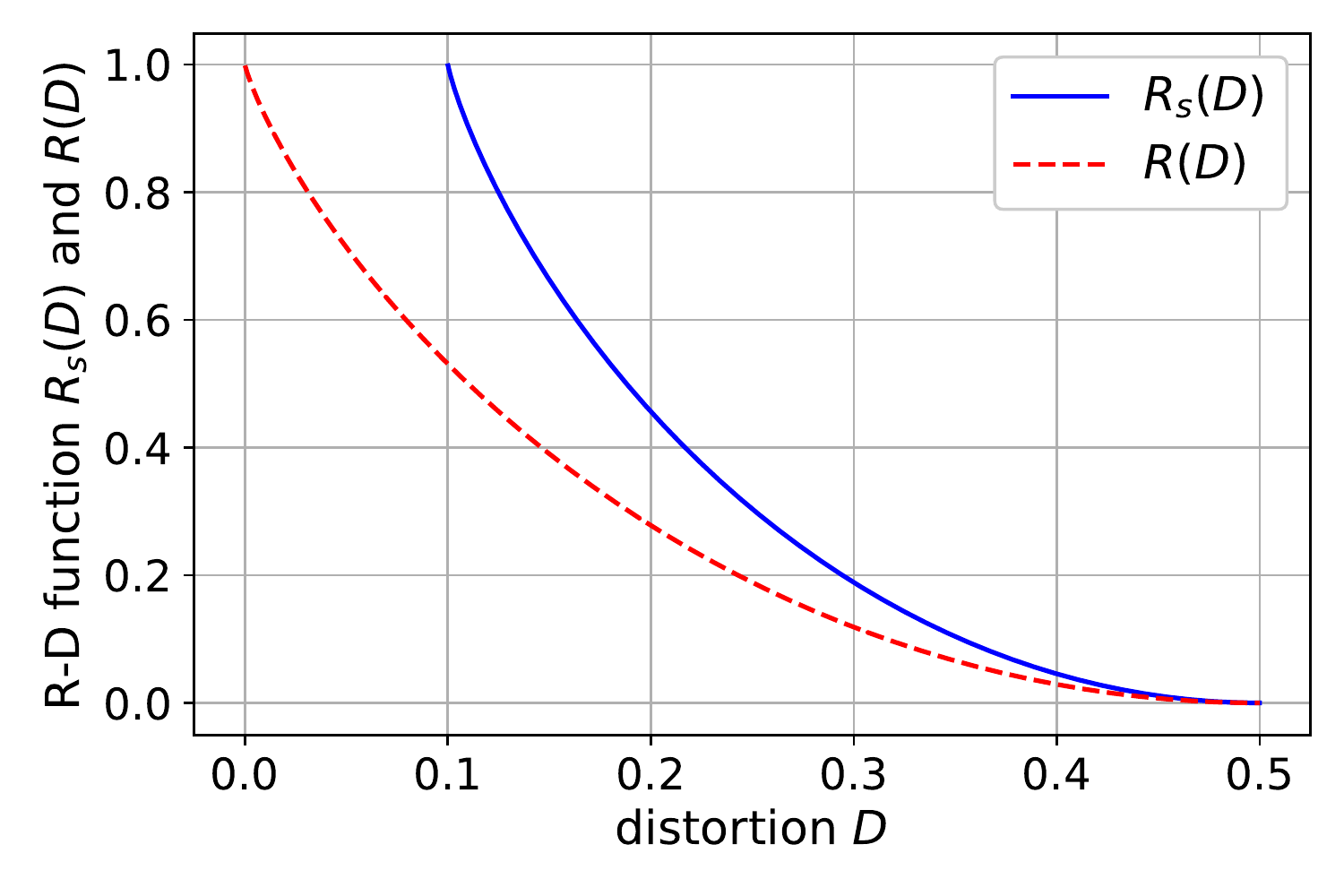}
			\caption{Comparison of rate-distortion functions $R_s(D)$ and $R(D)$ for $p=0.1$.}
			\label{fig_compare-RsD-RD}
		\end{figure}
	\end{remark}

	\begin{remark}
		We can imagine that $d_s'$ measures the distortion between the observation and reconstruction of semantic information. 
		Furthermore, it was shown in \cite{semantic-indirect-80,Liu-Zhang-Poor-ISIT21,Berger-book-71} that $d_s$ and $d_s'$ measure equivalent distortions, i.e.,
		\begin{align*}
		\bE d_s'(X_1,\hat{S}) &=\bE d_s(S,\hat{S})  \\
		\bE d_s'(X_1^n,\hat{S}^n) &=\bE d_s(S^n,\hat{S}^n).
		\end{align*}
		Then we can regard the system of compressing $X_1^n$ and reconstructing $\hat{S}^n$ as the ordinary rate-distortion problem with distortion measure $d_s'$. 
		Thus, $R_s(D_s)$ is equivalent to the ordinary rate-distortion function in \eqref{basic-RDfn} under distortion measure $d_s'$, which rigorously proves \eqref{RD-fn-semantic}. 
	\end{remark}

	\section{Case Study: Binary Sources}\label{section:example-binary}
	Assume $S$ and $X_1$ are doubly symmetric binary sources with distribution in \eqref{DSB-ditrbution-SX1}, 
	$X_2$ and $Y$ are both Bernoulli$(\frac{1}{2})$ sources. 
	The reconstructions are all binary, i.e., $\hat{\cX}_1=\hat{\cX}_2=\hat{\cS}=\{0,1\}$. 
	The distortion measures $d_1,d_2$, and $d_s$ are all assumed to be Hamming distortion. 
	We further assume that any two of $X_1$, $X_2$, and $Y$ are doubly symmetric binary distributed (c.f. \eqref{DSB-ditrbution}) with parameters $p_1,p_2$, and $p_3$, respectively. 
	Specifically, $(X_1,X_2)\sim \text{DSBS}(p_1)$, $(X_1,Y)\sim \text{DSBS}(p_2)$, $(X_2,Y)\sim \text{DSBS}(p_3)$. 
	
	Consider the following two examples that only differ in the source distributions.

	\subsection{Conditionally Independent Sources}
	Assume we have the Markov chain\footnotemark $X_1-Y-X_2$, i.e., $X_1$ and $X_2$ are independent conditioning on $Y$. 
	This assumption coincides with the intuitive understanding of $X_1$ and $X_2$ in \Cref{section:problemSetup}, 
	that the semantic feature can be independent with the background. 
	\footnotetext{Note that the Markov chain $X_1-Y-X_2$ indicates $p_1=p_2\star p_3\triangleq p_2(1-p_3)+p_3(1-p_2)$.}
	Then from \Cref{lemma:opt-separate-compression}, compressing $X_1^n$ and $X_2^n$ simultaneously is 
	the same as compressing them separately in terms of the optimal compression rate and distortions, 
	which implies the following theorem. 
	
	\begin{theorem}\label{thm:indep}
		The rate-distortion function for the above conditionally independent sources is given by 
		\begin{align*}
		R(D_1,D_2,D_s)	&= \big[h_b(p_3)-h_b(D_2)\big]\cdot \mathds{1}_{_{0\leq D_2\leq p_3}}   \nonumber \\
		&\hspace{-1.0cm} + \left[h_b(p_2)-h_b\big(\min\{D_1,D_s^0\}\big)\right]\cdot \mathds{1}_{_{0\leq \min\{D_1,D_s^0\}\leq p_2}}, 
		\end{align*}
		where $D_s^0=\frac{D_s-p}{1-2p}$ is defined in \eqref{def-Ds0}. 
	\end{theorem}
	\begin{proof}
		The rate-distortion function  in \Cref{thm:main} satisfies 
		\begin{align*}
		&R(D_1,D_2,D_s) \nonumber \\
		&= R_{\text{2d},X_1|Y}(D_1,D_s) + R_{X_2|Y}(D_2)  \\
		&= \min_{\substack{\bE d_1(X_1,\hat{X}_1)\leq D_1 \\ \bE d_s'(X_1,\hat{S})\leq D_s}} I(X_1;\hat{X}_1,\hat{S}|Y) + \hspace{-0.1cm} \min_{\bE d_2(X_2,\hat{X}_2)\leq D_2} I(X_2;\hat{X}_2|Y)  \\
		&= \left[h_b(p_2)-h_b\big(\min\{D_1,D_s^0\}\big)\right]\cdot \mathds{1}_{_{0\leq \min\{D_1,D_s^0\}\leq p_2}} \nonumber \\
		&\qquad + \big[h_b(p_3)-h_b(D_2)\big]\cdot \mathds{1}_{_{0\leq D_2\leq p_3}}
		\end{align*}
		where the last step follows from the rate-distortion functions in \eqref{basic-RDfn-SI-binary}, \eqref{basic-RDfn-2d-minD} and \eqref{RD-fn-semantic}. 
	\end{proof}
	

	\subsection{Correlated Sources}
	Similar to that in \cite{Kipnis-ITW15}, evaluating the rate-distortion function for the general correlated sources can be extremely difficult. 
	Thus, we assume the Markov chain $Y-X_1-X_2$ behind the intuition that the side information $Y$ can help more to the semantic related source $X_1$. 
	
	Without the conditional independence of $X_1$ and $X_2$, 
	the optimality of separate compression in \Cref{lemma:opt-separate-compression} may not hold. 
	The rate-distortion function $R(D_1,D_2,D_s)$ in \Cref{thm:main} can be calculated as follows. 
	Recall from \eqref{def-Ds0} that $D_s^0\triangleq \frac{D_s-p}{1-2p}$. 
	For simplicity, we consider only small distortions in the set  
	\begin{align}
	\cD_0 &= \big\{(D_1,D_2,D_s): 0\leq \min\{D_1, D_s^0\}\leq p_1p_2  \nonumber \\
	&\hspace{2.8cm}\text{ and }0\leq D_2\leq p_1\big\}.  \label{Correlated-distortion_region}
	\end{align}
	
	\begin{theorem}\label{thm:correlated}
		For $(D_1,D_2,D_s)\in\cD_0$, the rate-distortion function for the above correlated sources is 
		\begin{align}
		&R(D_1,D_2,D_s)  \nonumber \\
		&=h_b(p_1)+h_b(p_2) - h_b\left(\min\left\{D_1,D_s^0\right\}\right)-h_b(D_2),  \label{RD-fn-correlated}
		\end{align}
		where $D_s^0=\frac{D_s-p}{1-2p}$ is defined in \eqref{def-Ds0}. 
	\end{theorem}
	\begin{proof}
		The proof is given in Appendix~\ref{proof-thm:correlated}. 
	\end{proof}
	
	From the distribution of $(X_1,X_2)$ and $(X_1,Y)$, and the Markov chain $Y-X_1-X_2$, 
	it is easy to check that $(X_2,Y)$ is doubly symmetric distributed with parameter $p_3=p_1\star p_2\triangleq p_1(1-p_2)+p_2(1-p_1)$. 
	Then comparing \eqref{RD-fn-correlated} with the rate of separate compression, we have for $(D_1,D_2,D_s)\in\cD_0$ that 
	\begin{align}
	&R_{\text{2d},X_1|Y}(D_1,D_s) + R_{X_2|Y}(D_2)  \nonumber \\
	&= \min_{\substack{\bE d_1(X_1,\hat{X}_1)\leq D_1 \\ \bE d_s'(X_1,\hat{S})\leq D_s}} I(X_1;\hat{X}_1,\hat{S}|Y) + \min_{\bE d_2(X_2,\hat{X}_2)\leq D_2} I(X_2;\hat{X}_2|Y)  \nonumber \\
	&= \big[h_b(p_2)-h_b(\min\{D_1,D_s^0\})\big] + \big[h_b(p_1\star p_2)-h_b(D_2)\big] \nonumber \\
	&\geq R(D_1,D_2,D_s),  \label{Correlated-compare-rates}
	\end{align}
	where the last inequality follows from the fact that $h_b(\cdot)$ is increasing in $[0,0.5]$ and $p_1\star p_2\geq p_1$ for $0\leq p_1\leq0.5$. 
	
	Note that the equality in \eqref{Correlated-compare-rates} holds only for $p_1=0.5$, which together with $Y-X_1-X_2$ imply the Markov chain $X_1-Y-X_2$ in \Cref{lemma:opt-separate-compression}. 
	Then the problem reduces to that in \Cref{thm:indep}. 
	For $p_1<0.5$, simultaneously compressing $X_1$ and $X_2$ is strictly better than separate compression. 
	
	%

	\section{Case Study: Binary Classification of Integers}\label{section:example-integer}
	Consider classification integers into even and odd. 
	Let $X_1$ be uniformly distributed over $\cX_1=[1:N]$ with $N\geq 4$ being even. 
	The semantic information $S$ is a binary random variable probabilistically indicates whether $X_1$ is even or odd. 
	The transition probability can be defined by BSC in Fig.~\ref{fig_BSC-even-odd}, which is similar to that in Fig.~\ref{fig_BSC} by replacing the value of $X_1$  with ``even" and ``odd". 
	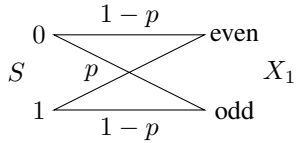
\begin{figure}[!h]
		\centering
		\begin{tikzpicture}
		\node at (0,0) {$S$};
		\node at (0.3,0.5) {$0$};
		\node at (0.3,-0.5) {$1$};
		
		\draw (0.5,-0.5)--(2.5,-0.5)--(0.5,0.5)--(2.5,0.5)--cycle;
		
		\node at (3.5,0) {$X_1$};
		\node at (2.9,0.5) {even};
		\node at (2.9,-0.5) {odd};
		
		\node at (1.5,0.75) {$1-p$};
		\node at (1.5,-0.75) {$1-p$};
		\node at (1.0,0) {$p$};
		\end{tikzpicture}
		\caption{Transition probability from the binary semantic information $S$ to the integer $X_1$: in terms of BSC.}
		\label{fig_BSC-even-odd}
	\end{figure}
	The binary side information $Y$ is correlated with $X_1$ also indicating its odevity (even/odd) similar to Fig.~\ref{fig_BSC-even-odd} with parameter $p_2$. 
	Assume the Markov chain $X_1-Y-X_2$ holds, and the Bernoulli($\frac{1}{2}$) source $X_2$ is independent with~$Y$. 
	We can verify that $X_2$ is independent with $(X_1,Y)$. 
	By \Cref{lemma:opt-separate-compression}, compressing $X_1^n$ and $X_2^n$ simultaneously is the same as compressing them separately. 
	For simplicity, we consider only small distortions in the set  
	\begin{align}
	\cD_1 &= \bigg\{(D_1,D_2,D_s): 0\leq \min\{D_1, D_s^0\}\leq \frac{2(N-1)p_2}{N}   \nonumber \\
	&\qquad\qquad\qquad\qquad  \text{ and }0\leq D_2\leq 0.5\bigg\}.  \label{Classify-distortion_region}
	\end{align}
	
	\begin{theorem}\label{thm:classify}
		For $(D_1,D_2,D_s)\in\cD_1$, the rate-distortion function for integer classification is 
		\begin{align*}
		R(D_1,D_2,D_s)&= \big[h_b(p_2)+\log(N/2) - h_b(\min\{D_1,D_s^0\})  \nonumber \\
		&\qquad -D_1\log(N-1)\big] + \big[1-h_b(D_2)\big], 
		\end{align*}
		where $D_s^0=\frac{D_s-p}{1-2p}$ is defined in \eqref{def-Ds0}. 
	\end{theorem}
	\begin{proof}
		The proof is given in Appendix~\ref{proof-thm:classify}. 
	\end{proof}

	\section{Case Study: Gaussian Sources}\label{section:example-Gaussian}
	Assume $S$ and $X_1$ are jointly Gaussian sources with zero mean and covariance matrix
	\begin{equation}
	\begin{bmatrix}
	\sigma_S & \sigma_{SX_1} \\
	\sigma_{SX_1} & \sigma_{X_1}
	\end{bmatrix}.
	\end{equation}
	Similarly, we assume the Markov chain $X_1-Y-X_2$, where $X_2$ and $Y$ are jointly Gaussian sources with zero mean and covariance matrix
	\begin{equation}
	\begin{bmatrix}
	\sigma_{X_2} & \sigma_{X_2Y} \\
	\sigma_{X_2Y} & \sigma_{Y}
	\end{bmatrix}. 
	\end{equation}
	Thus $X_1$ is conditionally independent of $X_2$ given $Y$. 
	Let the covariance of $X_1$ and $Y$ be $\sigma_{X_1Y}$. 
	The reconstructions are real scalars, i.e., $\hat{\mathcal{X}}_1=\hat{\mathcal{X}}_2=\hat{\mathcal{S}}=\mathbb{R}$. The distortion metrics are squared error.
	
	We see from Lemma 3 that compressing $X_1^n$ and $X_2^n$ simultaneously is the same as compressing them separately in terms of the optimal compression rate and distortions. 
	Then we have the following theorem.
	
	\begin{theorem}\label{thm:Gaussian}
		For the Gaussian sources, if the Markov chain $X_1-Y-X_2$ holds, the rate-distortion function is
		\begin{align}
		&R(D_1,D_2,D_s)=\frac{1}{2}\left(\log\frac{\sigma_{X_2}-\frac{\sigma_{X_2Y}^2}{\sigma_{Y}}}{D_2}\right)^+ + \nonumber  \\ 
		&\frac{1}{2}\left[\log\max\left(\frac{\sigma_{X_1}-\frac{\sigma_{X_1Y}^2}{\sigma_{Y}}}{D_1}, \frac{\sigma_{SX_1}^2\left(\sigma_{X_1}-\frac{\sigma_{X_1Y}^2}{\sigma_{Y}}\right)}{\sigma_{X_1}^2\left(D_s-\mathrm{mmse}\right)}\right)\right]^+,  \label{eqn:Gaussian SI}
		\end{align}
		where $\mathrm{mmse}$ is the minimum mean squared error for estimating $S$ from $X_1$, given by 
		\begin{equation}\label{eqn:mmse_def}
			\mathrm{mmse}= \sigma_S-\frac{\sigma_{SX_1}^2}{\sigma_{X_1}}.
		\end{equation}
	\end{theorem} 
	\begin{proof}
		The proof is given in Appendix~\ref{proof-thm:Gaussian}.
	\end{proof}

	\section{Numerical Results}\label{section:plots}
	In this section, we plot the rate-distortion curves evaluated in the previous sections. 
	\subsection{Correlated Binary Sources}
	\begin{figure}[!t]
		\centering
		\includegraphics[scale=0.5]{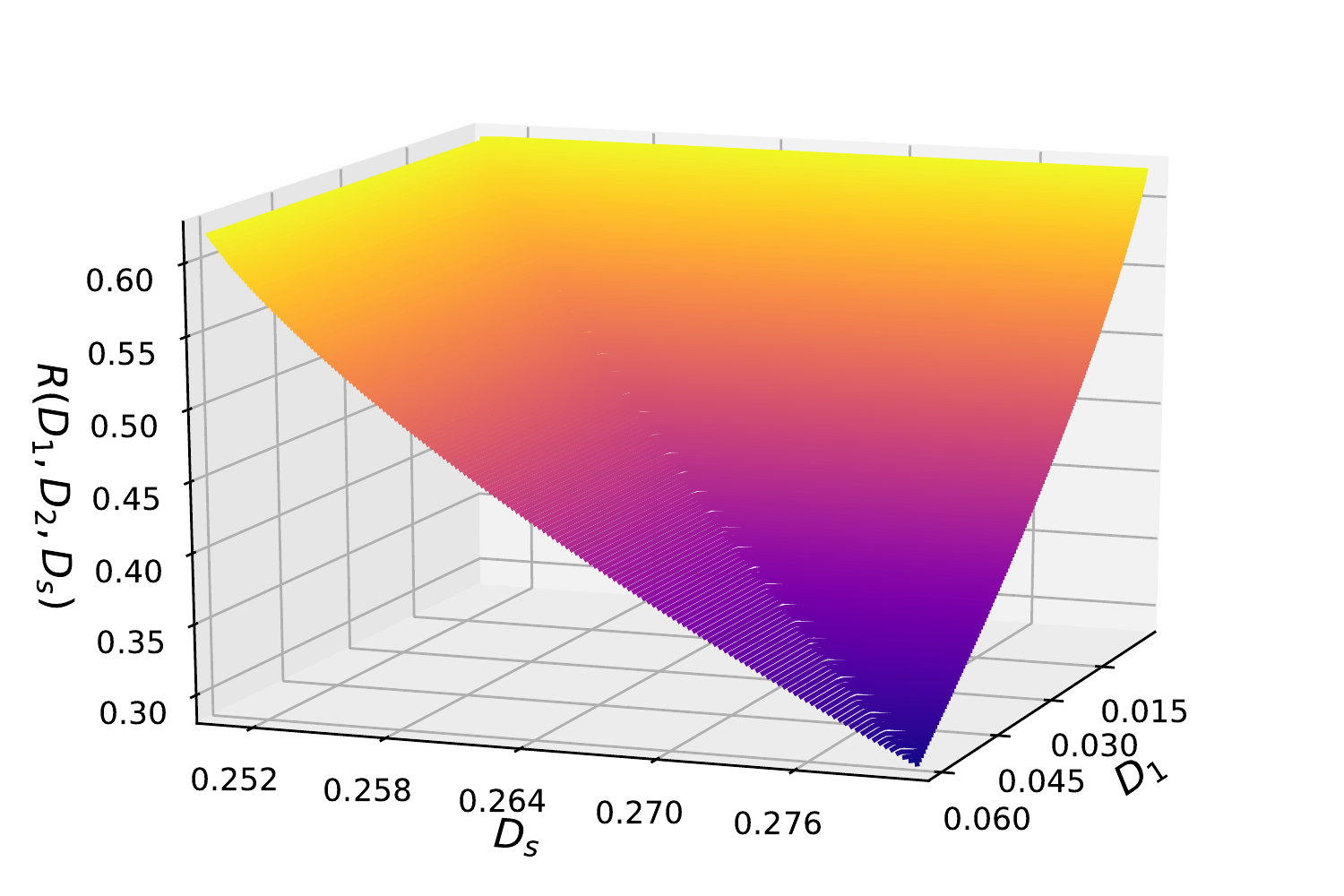}
		\caption{Rate-distortion function for correlated binary sources: $p=p_1=p_2=0.25$ and $D_2=0.5$. }
		\label{fig_binary-D1-Ds}
	\end{figure}
	
	\begin{figure}[!t]
		\centering
		\begin{subfigure}{.23\textwidth}
			\centering
			\includegraphics[scale=0.29]{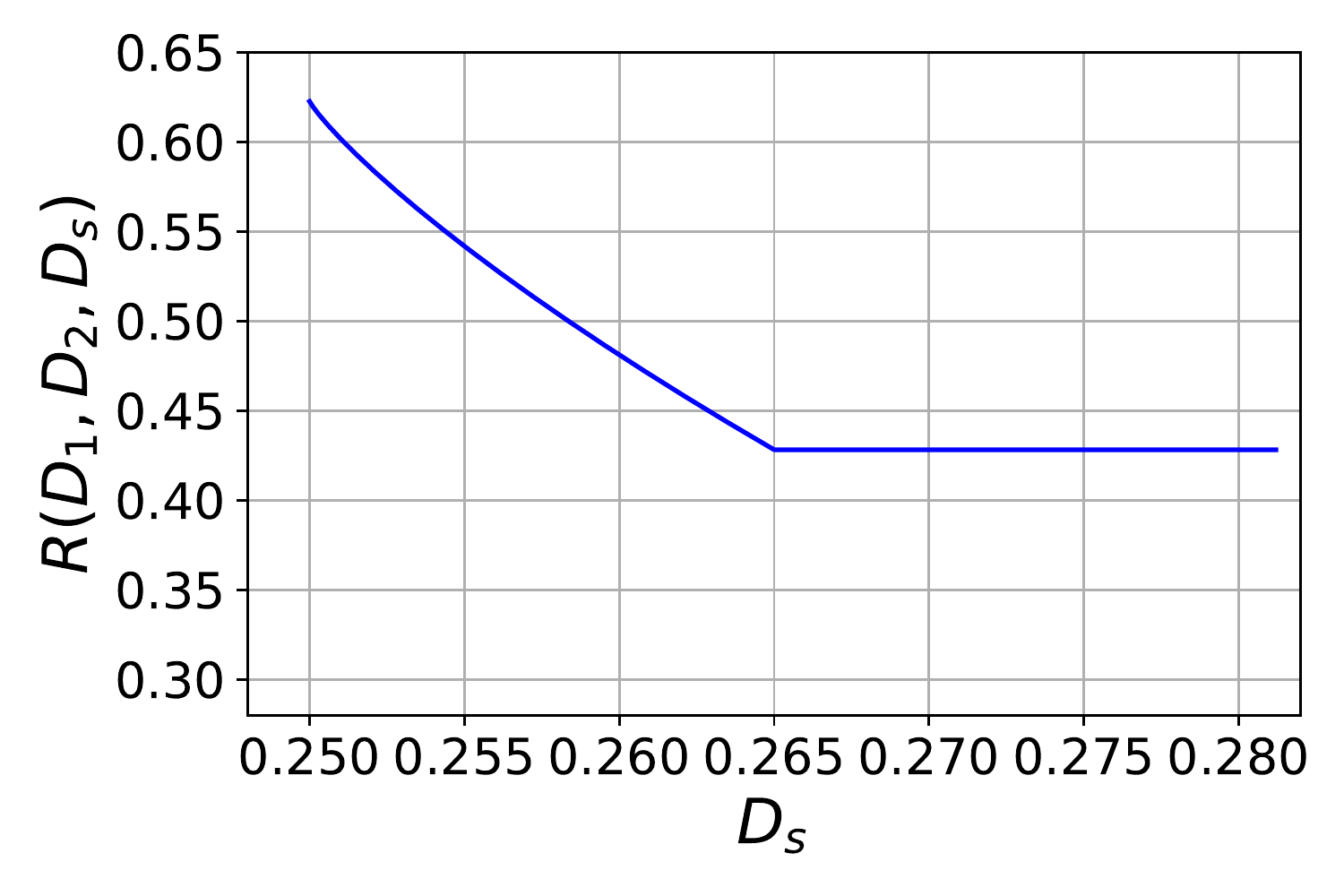}
			\caption{$D_1=0.03$}
			\label{fig_binary-Ds-1}
		\end{subfigure}
		~
		\begin{subfigure}{.23\textwidth}
			\centering
			\includegraphics[scale=0.29]{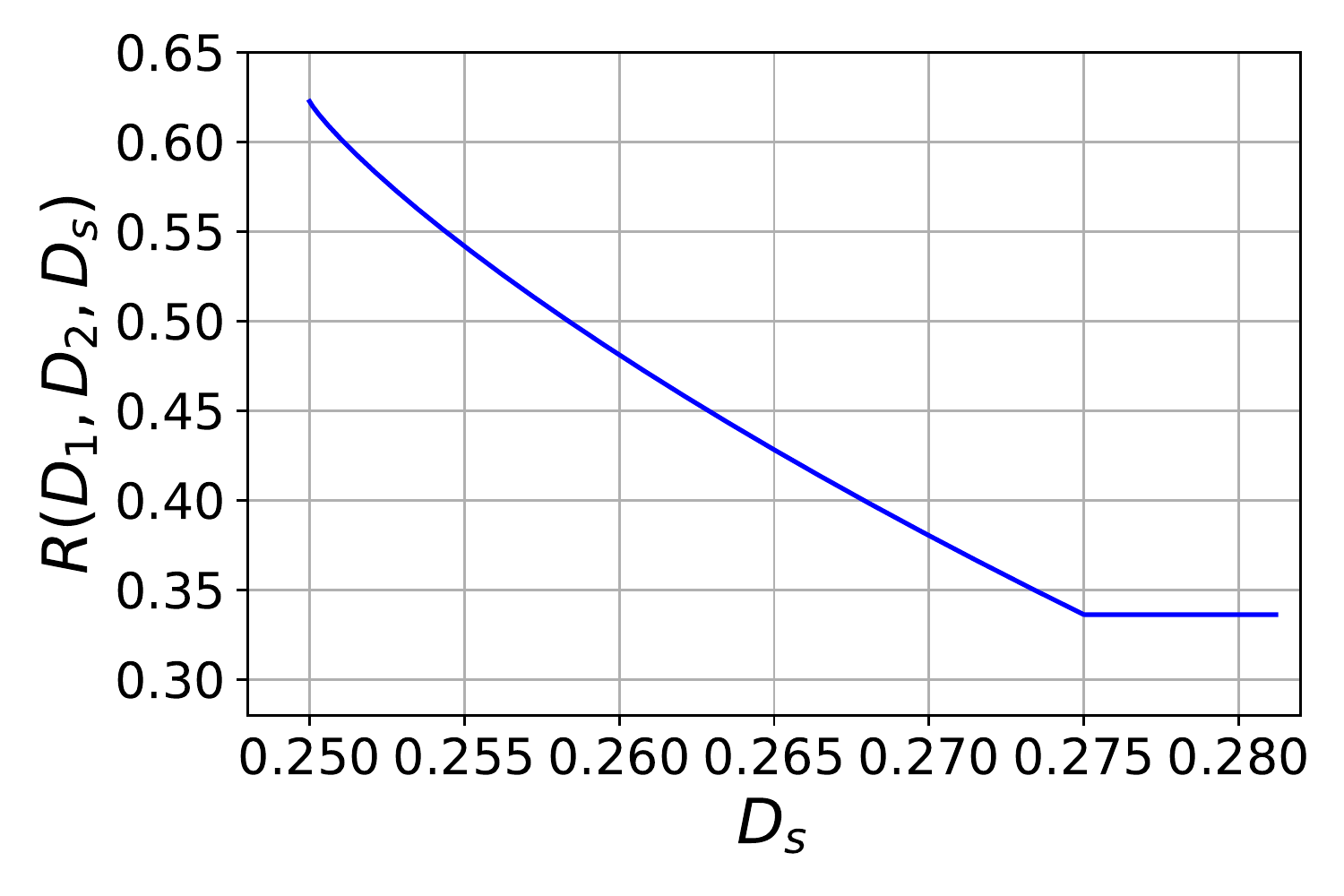}
			\caption{$D_1=0.05$}
			\label{fig_binary-Ds-2}
		\end{subfigure}
		\caption{Rate-distortion function for correlated binary sources: $p=p_1=p_2=0.25$ and $D_2=0.5$.}
		\label{fig_binary-Ds}
	\end{figure}
	Consider the rate-distortion function for correlated binary sources in \Cref{thm:correlated} with $p=p_1=p_2=0.25$ and $D_2=0.5$. 
	Fig.~\ref{fig_binary-D1-Ds} shows the 3-D plot of the optimal tradeoff between the coding rate $R$ and distortions $(D_1,D_s)$. 
	We can see that the rate-distortion function is decreasing and convex in $(D_1,D_s)$ for distortions in $\cD_0$ (c.f.~\eqref{Correlated-distortion_region}). 
	
	The truncated curves with $D_1=0.03$ and $D_1=0.05$ are shown in Fig.~\ref{fig_binary-Ds}. 
	We see that the rate is decreasing in $D_s$ until it achieves the minimum rate, which is determined by $D_1$. 
	Similar curves can also be obtained by truncating with some constant $D_s$. 
	

	\subsection{Binary Classification of Integers}
	The rate-distortion function for integer classification in \Cref{thm:classify} with $p=p_1=p_2=0.25$, $D_2=0.5$, and $N=8$ is illustrated in Fig.~\ref{fig_classify}. 
	Note that in both \Cref{thm:correlated} and \Cref{thm:classify}, $D_2=0.5$ indicates that $X_2$ can be recovered by random guessing, 
	which further implies that $X_2$ can also be regarded as side information at both sides. 
	
	Comparing the rates along the $D_1$ and $D_s$ axis in Fig.~\ref{fig_classify}, 
	we see that recovering only the semantic information can reduce the coding rate comparing to recovering the source message.
	
	Comparing Fig.~\ref{fig_classify} with Fig.~\ref{fig_binary-D1-Ds}, we see that the rate for integer classification decreases faster as $D_1$ increases (which is clearer at the minimum of $D_s=p$). 
	This implies that $D_1$ is more dominant (to determine the rate) here, 
	which is intuitive since the integer source has a larger alphabet and recovering it with different distortions requires a larger range of rates.
	
	\begin{figure}[!t]
		\centering
		\includegraphics[scale=0.5]{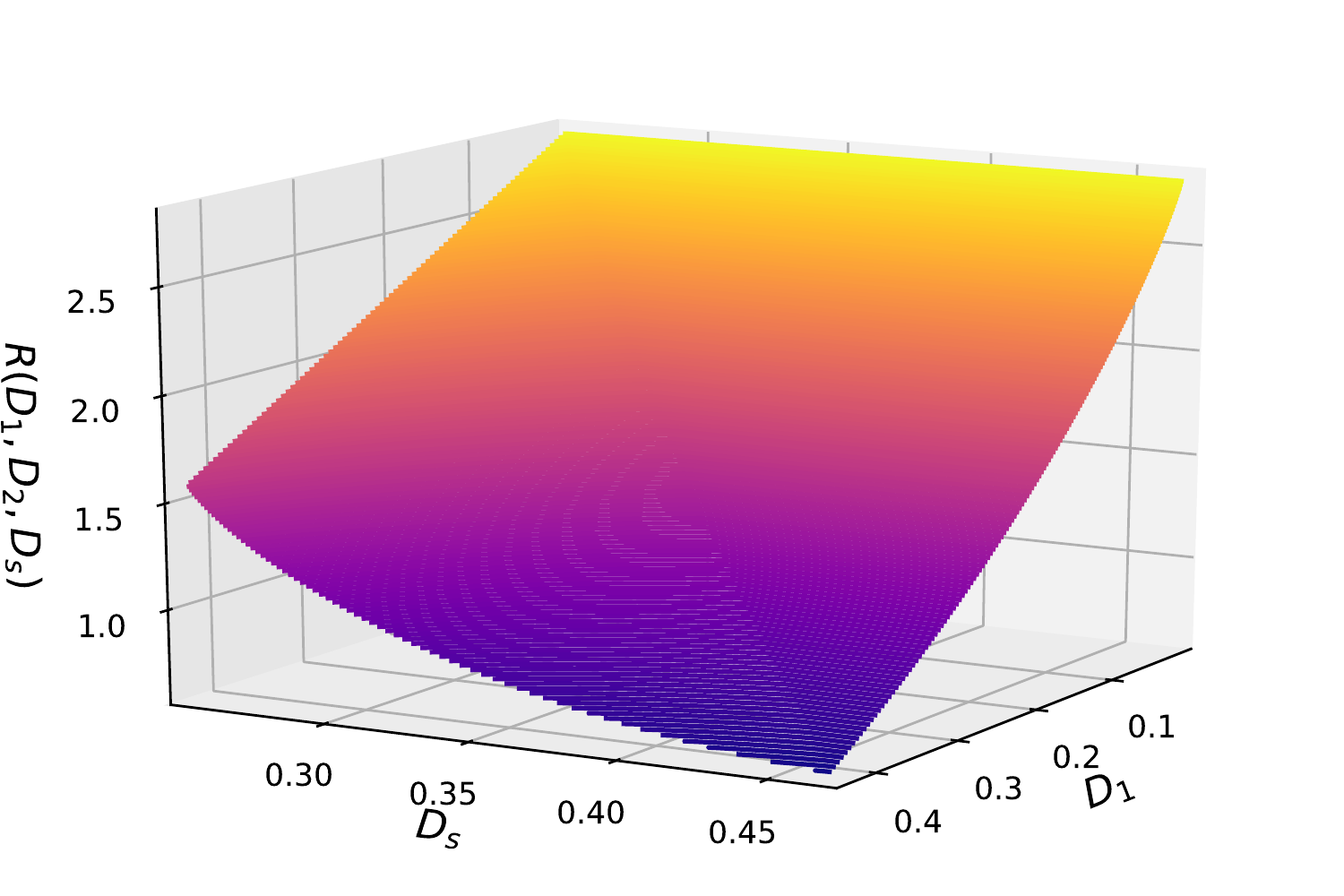}
		\caption{Rate-distortion function for integer classification: $p=p_1=p_2=0.25$, $D_2=0.5$, and $N=8$. }
		\label{fig_classify}
	\end{figure}

	\subsection{Gaussian Sources}
	Consider the rate-distortion function for Gaussian sources in Theorem \ref{thm:Gaussian}. 
	Let all of the variances be $2$, all of the covariances be $1$, and $D_2=1$. 
	
	The 3-D plot of the optimal tradeoff between the coding rate $R$ and distortions $(D_1, D_s)$ is illustrated in Fig.~\ref{fig:3D}. 
	We can see that the rate-distortion function is decreasing and convex in $(D_1, D_s)$. 
	The minimum rate is equal to 
	$$R_{X_2|Y}(D_2)=\frac{1}{2}\left(\log\frac{\sigma_{X_2}-\frac{\sigma_{X_2Y}^2}{\sigma_{Y}}}{D_2}\right)^+=0.20. $$
	
	The contour plot of the rate-distortion function is shown in Fig.~\ref{fig:2D}. 
	The slanted line denotes the situations that $R_{X_1|Y}(D_1)=R_{S|Y}(D_s)$. 
	We see that when $D_1$ is more dominant (the region above the slanted line), 
	the rate only needs to meet the distortion constraint to reconstruct $X_1$. 
	On the contrary, when $D_s$ is more dominant (the region below the slanted line), 
	the rate only needs to meet the distortion constraint to reconstruct $S$. 
	
	\begin{figure}[!thp]
		\centering
		\includegraphics[scale=0.5]{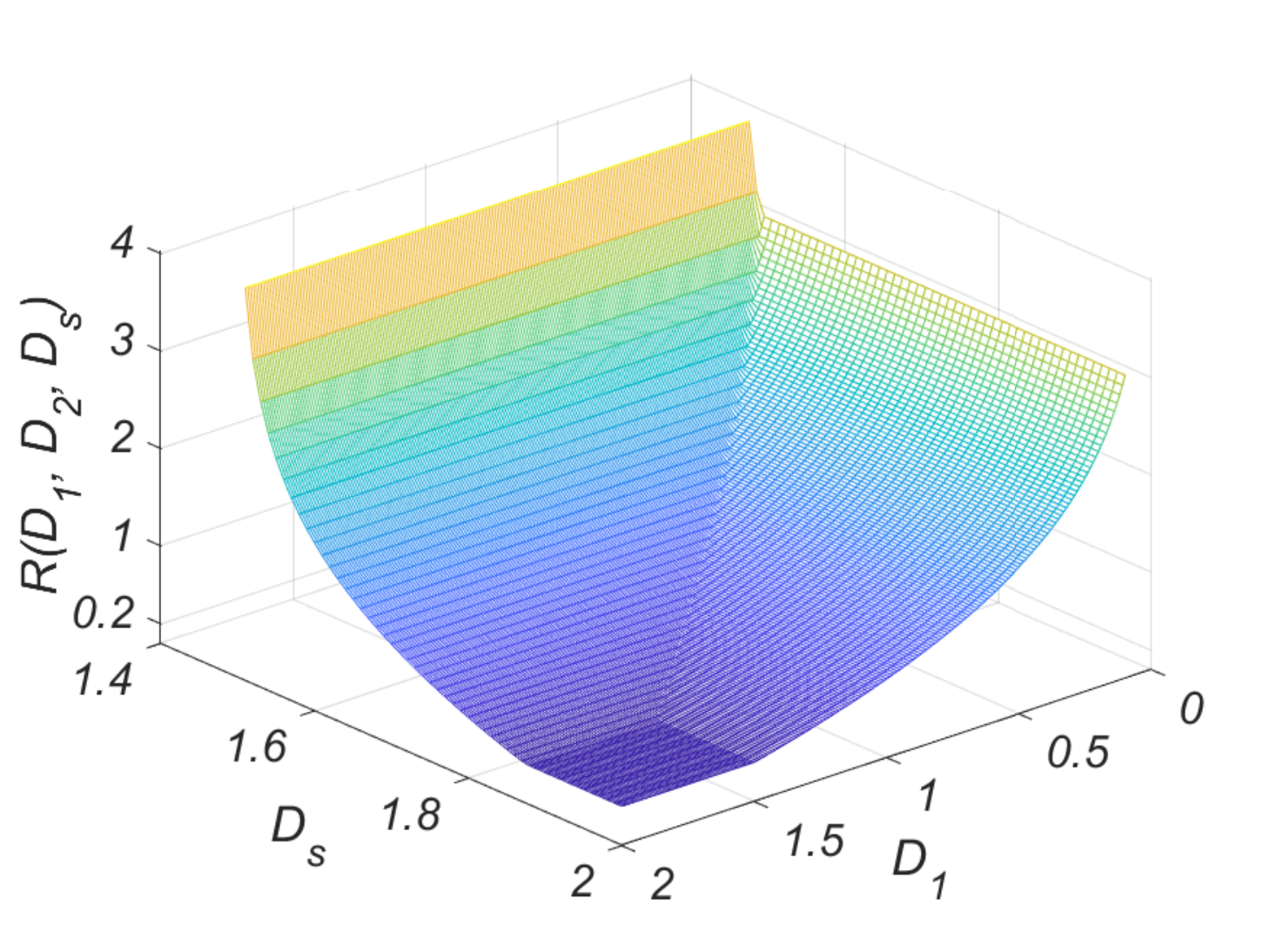}
		\caption{Rate-distortion function for Gaussian sources: $\sigma_S=\sigma_{X_1}=\sigma_{X_2}=\sigma_{Y}=1$, $\sigma_{SX_1}=\sigma_{X_1Y}=\sigma_{X_2Y}=2$, and  $D_2 = 1$.}
		\label{fig:3D}
	\end{figure}
	
	\begin{figure}[!thp]
		\centering
		\includegraphics[scale=0.3]{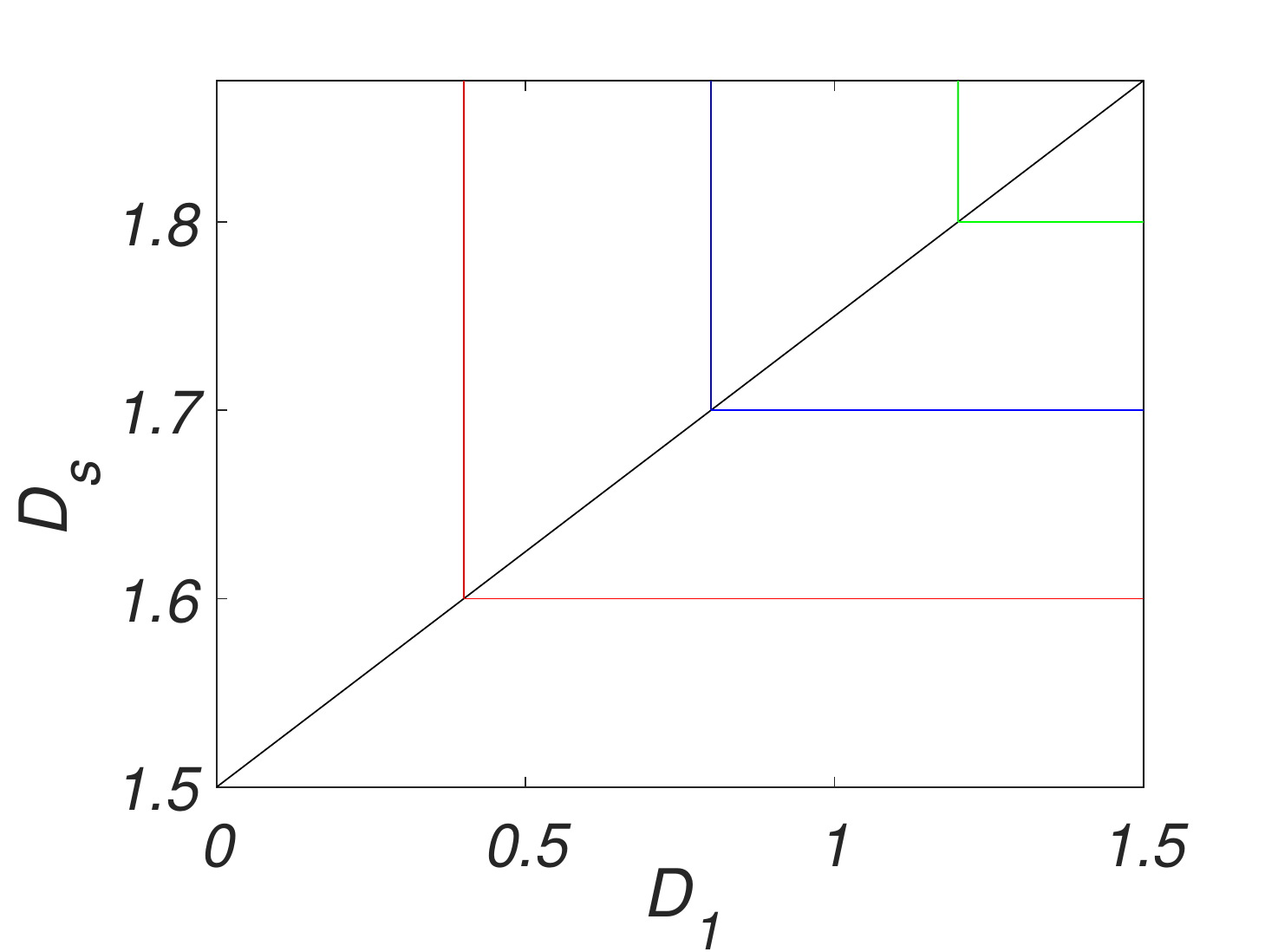}
		\caption{The contour plot of the rate-distortion function.}
		\label{fig:2D}
	\end{figure}

	\section{Conclusion}\label{section:conclusion}
	In this paper, we studied the semantic rate-distortion problem with side information motivated by task-oriented video compression. 
	The general rate-distortion function was characterized. 
	We also evaluated several cases with specific sources and distortion measures. 
	It is more desirable to derive the rate-distortion function for real video sources, 
	which is more challenging due to the high complexity of real source models and choice of meaningful distortion measures. 
	This part of work is now under investigation.

	\appendices
	\section{Proof of \Cref{thm:main}} \label{proof-thm:main}
	The achievability part is a straightforward extension of the joint typicality coding scheme for lossy source coding. 
	We simply present the coding ideas and analysis as follows. 
	Fix the conditional pmf $p(\hat{x}_1,\hat{x}_2,\hat{s}|x_1,x_2,y)$ such that the distortion constraints are satisfied, 
	$\bE d_1(X_1,\hat{X}_1)\leq D_1$, $\bE d_2(X_2,\hat{X}_2)\leq D_2$, and $\bE d_s'(X_1,\hat{S})\leq D_s$. 
	Let $p(\hat{x}_1,\hat{x}_2,\hat{s}|y)= \sum_{x_1,x_2}p(x_1,x_2|y)p(\hat{x}_1,\hat{x}_2,\hat{s}|x_1,x_2,y)$. 
	Randomly and independently generate $2^{nR}$ sequence triples $(\hat{x}_1^n,\hat{x}_2^n,\hat{s}^n)$ indexed by $m\in[1:2^{nR}]$, each according to $p(\hat{x}_1,\hat{x}_2,\hat{s}|y)$. 
	The whole codebook $\cC$, consisting of these sequence triples, is revealed to both the encoder and decoder. 
	When observing the source messages $(x_1^n,x_2^n,y^n)$, find an index $m$ such that its indexing sequence $(\hat{x}_1^n,\hat{x}_2^n,\hat{s}^n)$ satisfies $(x_1^n,x_2^n,y^n,\hat{x}_1^n,\hat{x}_2^n,\hat{s}^n)\in\cT_{\epsilon}^n$. 
	If there is more than one such index, randomly choose one of them; if there is no such index, set $m=1$. 
	Upon receiving the index $m$, the decoder reconstruct the messages and inference by choosing the codeword $(\hat{x}_1^n,\hat{x}_2^n,\hat{s}^n)$ indexed by $m$. 
	By law of large numbers, the source sequences are joint typical with probability 1 as $n\rightarrow\infty$. 
	Then we define the ``encoding error" event as 
	\begin{equation}
	\mathcal{E}=\left\{\left(X_1^n,X_2^n,Y^n,\hat{X}_1^n,\hat{X}_2^n,\hat{S}^n\right)\notin\cT_{\epsilon}^n, ~\forall m\in\left[1:2^{nR}\right]\right\}. 
	\end{equation}
	Then we can bound the error probability as follows 
	\begin{align*}
	&P(\mathcal{E}) \nonumber \\
	&= P\left\{\left(x_1^n,x_2^n,y^n,\hat{X}_1^n,\hat{X}_2^n,\hat{S}^n\right)\notin\cT_{\epsilon}^n, ~\forall m\in\left[1:2^{nR}\right]\right\}  \\
	&=\prod_{m=1}^{2^{nR}}P\left\{\left(x_1^n,x_2^n,y^n,\hat{X}_1^n,\hat{X}_2^n,\hat{S}^n\right)\notin\cT_{\epsilon}^n\right\}  \\
	&=\left(1-P\left\{\left(x_1^n,x_2^n,y^n,\hat{X}_1^n,\hat{X}_2^n,\hat{S}^n\right)\in\cT_{\epsilon}^n\right\}\right)^{2^{nR}}  \\
	&\leq \sum_{(x_1^n,x_2^n,y^n)\in\cT_{\epsilon}^n} \Bigg[p(x_1^n,x_2^n,y^n)\cdot \nonumber \\ 
	&\qquad \left(1-2^{-n[I(X_1,X_2;\hat{X}_1,\hat{X}_2,\hat{S}|Y)+\delta(\epsilon)]}\right)^{2^{nR}}\Bigg]  \\
	&\leq \exp\left(-2^{n[R-I(X_1,X_2;\hat{X}_1,\hat{X}_2,\hat{S}|Y)-\delta(\epsilon)]}\right), 
	\end{align*}
	where $\delta(\epsilon)\rightarrow0$ as $n\rightarrow\infty$, 
	the first inequality follows from the joint typicality lemma in \cite{NetworkIT-book}, 
	and the last inequality follows from the fact that $(1-z)^t\leq \exp(-tz)$ for $z\in[0,1]$ and $t\geq 0$. 
	We see that $P(\mathcal{E})\rightarrow0$ as $n\rightarrow\infty$ if $R> I(X_1,X_2;\hat{X}_1,\hat{X}_2,\hat{S}|Y)+\delta(\epsilon)$. 
	If the error event does not happen, i.e., the reconstruction is joint typical with the source sequences, 
	then from the distortion constraints assumed for the conditional pmf, the expected distortions can achieve $D_1,D_2,$ and $D_s$, respectively. 
	This proves the achievability. 
	
	Define $R_I(D_1,D_2,D_s)$ as the rate-distortion function characterized by \Cref{thm:main}, 
	For the converse part, we show that 
	\begin{align}
	nR&\geq H(W) \geq H(W|Y^n) \geq I(X_1^n,X_2^n;W|Y^n)  \nonumber \\
	&\geq I(X_1^n,X_2^n;\hat{X}_1^n,\hat{X}_2^n,\hat{S}^n|Y^n)  \nonumber \\
	&= I(X_1^n,X_2^n,Y^n;\hat{X}_1^n,\hat{X}_2^n,\hat{S}^n) - I(Y^n;\hat{X}_1^n,\hat{X}_2^n,\hat{S}^n) \nonumber \\
	&=\sum_{i=1}^n \left[I(X_{1,i},X_{2,i},Y_i;\hat{X}_1^n,\hat{X}_2^n,\hat{S}^n|X_{1,1}^{i-1},X_{2,1}^{i-1},Y_1^{i-1})\right. \nonumber \\
	&\qquad \left.- I(Y_i;\hat{X}_1^n,\hat{X}_2^n,\hat{S}^n|Y_1^{i-1})\right]   \nonumber \\
	&=\sum_{i=1}^n \left[I(X_{1,i},X_{2,i},Y_i;\hat{X}_1^n,\hat{X}_2^n,\hat{S}^n,X_{1,1}^{i-1},X_{2,1}^{i-1},Y_1^{i-1}) \right. \nonumber \\
	&\qquad \left.- I(Y_i;\hat{X}_1^n,\hat{X}_2^n,\hat{S}^n,Y_1^{i-1})\right]   \nonumber \\
	&=\sum_{i=1}^n \left[I(X_{1,i},X_{2,i};\hat{X}_{1,i},\hat{X}_{2,i},\hat{S}_{i}|Y_i) \right. \nonumber \\
	&\qquad  + I(X_{1,i},X_{2,i},Y_i;\hat{X}_{1,1}^{i-1},\hat{X}_{1,i+1}^{n},\hat{X}_{2,1}^{i-1},\hat{X}_{2,i+1}^{n},  \nonumber \\
	&\qquad \qquad \hat{S}_{1}^{i-1},\hat{S}_{i+1}^{n},X_{1,1}^{i-1},X_{2,1}^{i-1},Y_1^{i-1}|\hat{X}_{1,i},\hat{X}_{2,i},\hat{S}_{i})  \nonumber  \\
	&\qquad  - I(Y_i;\hat{X}_{1,1}^{i-1},\hat{X}_{1,i+1}^{n},\hat{X}_{2,1}^{i-1},\hat{X}_{2,i+1}^{n},  \nonumber \\
	&\qquad \qquad \left.\hat{S}_{1}^{i-1},\hat{S}_{i+1}^{n},Y_1^{i-1}|\hat{X}_{1,i},\hat{X}_{2,i},\hat{S}_{i})\right]   \nonumber \\
	&=\sum_{i=1}^n \left[I(X_{1,i},X_{2,i};\hat{X}_{1,i},\hat{X}_{2,i},\hat{S}_{i}|Y_i) \right. \nonumber \\
	&\qquad + I(X_{1,i},X_{2,i},Y_i;X_{1,1}^{i-1},X_{2,1}^{i-1}|\hat{X}_1^n,\hat{X}_2^n,\hat{S}^n,Y_1^{i})  \nonumber \\
	&\qquad + I(X_{1,i},X_{2,i};\hat{X}_{1,1}^{i-1},\hat{X}_{1,i+1}^{n},\hat{X}_{2,1}^{i-1},\hat{X}_{2,i+1}^{n},  \nonumber \\
	&\qquad \qquad \left. \hat{S}_{1}^{i-1},\hat{S}_{i+1}^{n},Y_1^{i-1}|\hat{X}_{1,i},\hat{X}_{2,i},\hat{S}_{i},Y_i)\right]  \nonumber \\
	&\geq \sum_{i=1}^n I(X_{1,i},X_{2,i};\hat{X}_{1,i},\hat{X}_{2,i},\hat{S}_{i}|Y_i)  \label{converse-3} \\
	&\geq \sum_{i=1}^n R_I\left(\bE d_1(X_{1,i},\hat{X}_{1,i}), \bE d_2(X_{2,i},\hat{X}_{2,i}), \bE d_s'(X_{1,i},\hat{S}_i)\right)  \label{converse-4} \\
	&\geq n R_I\left(\bE d_1(X_1^n,\hat{X}_1^n), \bE d_2(X_2^n,\hat{X}_2^n), \bE d_s'(X_1^n,\hat{S}^n)\right) \label{converse-5} \\
	&\geq n R_I\left(\bE d_1(X_1^n,\hat{X}_1^n), \bE d_2(X_2^n,\hat{X}_2^n), \bE d_s(S^n,\hat{S}^n)\right) \label{converse-6} \\
	&\geq nR_I(D_1,D_2,D_s), 
	\end{align}
	where \eqref{converse-3} follows from the nonnegativity of mutual information, 
	\eqref{converse-4} follows from the definition of $R_I(D_1,D_2,D_s)$, 
	\eqref{converse-5} follows from the convexity of $R_I(D_1,D_2,D_s)$, 
	\eqref{converse-6} follows from $\bE d_s'(X_1^n,\hat{S}^n)=\bE d_s(S^n,\hat{S}^n)$ which is proved in \cite{Liu-Zhang-Poor-ISIT21}, 
	and the last inequality follows from the non-increasing property of $R_I(D_1,D_2,D_s)$. 
	This completes the converse proof.

	\section{Proof of \Cref{lemma:opt-separate-compression}}\label{proof-lemma:opt-separate-compression}
	The Markov chain $X_1-Y-X_2$ indicates that 
	\begin{equation}
	H(X_2|X_1,Y)=H(X_2|Y).  \label{sep_prf_indep-X2-X1Y}
	\end{equation}
	Then the mutual information in \eqref{RD-fn_rate} can be bounded by 
	\begin{align*}
	&I(X_1,X_2;\hat{X}_1,\hat{X}_2,\hat{S}|Y)  \nonumber \\
	&= H(X_1,X_2|Y)-H(X_1,X_2|\hat{X}_1,\hat{X}_2,\hat{S},Y)  \\
	&= H(X_1|Y)+H(X_2|Y)-H(X_1|\hat{X}_1,\hat{X}_2,\hat{S},Y) \nonumber \\
	&\qquad -H(X_2|X_1,\hat{X}_1,\hat{X}_2,\hat{S},Y)  \\
	&\geq H(X_1|Y)+H(X_2|Y) \! - \! H(X_1|\hat{X}_1,\hat{S},Y) \! -\! H(X_2|\hat{X}_2,Y)  \\
	&= I(X_1;\hat{X}_1,\hat{S}|Y) + I(X_2;\hat{X}_2|Y),
	\end{align*}
	where the inequality follows from the fact that conditioning does not increase entropy. 
	Now, we have 
	\begin{align*}
	&R(D_1,D_2,D_s)  \nonumber \\
	&= \min_{\substack{p(\hat{x}_1,\hat{x}_2,\hat{s}|x_1,x_2,y)\\ \bE d_1(X_1,\hat{X}_1)\leq D_1 \\ \bE d_2(X_2,\hat{X}_2)\leq D_2 \\ \bE d_s'(X_1,\hat{S})\leq D_s}} I(X_1,X_2;\hat{X}_1,\hat{X}_2,\hat{S}|Y)  \\
	&\geq \min_{\substack{p(\hat{x}_1,\hat{x}_2,\hat{s}|x_1,x_2,y)\\ \bE d_1(X_1,\hat{X}_1)\leq D_1 \\ \bE d_2(X_2,\hat{X}_2)\leq D_2 \\ \bE d_s'(X_1,\hat{S})\leq D_s}} \left[I(X_1;\hat{X}_1,\hat{S}|Y) + I(X_2;\hat{X}_2|Y)\right]  \\
	&= \min_{\substack{p(\hat{x}_1,\hat{s}|x_1,y)\\ \bE d_1(X_1,\hat{X}_1)\leq D_1 \\ \bE d_s'(X_1,\hat{S})\leq D_s}} I(X_1;\hat{X}_1,\hat{S}|Y) \nonumber \\
	&\qquad + \min_{\substack{p(\hat{x}_2|x_2,y)\\ \bE d_2(X_2,\hat{X}_2)\leq D_2}} I(X_2;\hat{X}_2|Y)  \\
	&= R_{\text{2d},X_1|Y}(D_1,D_s) + R_{X_2|Y}(D_2). 
	\end{align*}
	For the other direction, we show that the rate-distortion quadruple $\big(R_{\text{2d},X_1|Y}(D_1,D_s) + R_{X_2|Y}(D_2), D_1,D_2,D_s\big)$ is achievable. 
	To see this, let $p^*(\hat{x}_1,\hat{s}|x_1,y)$ and $p^*(\hat{x}_2|x_2,y)$ be the optimal distributions 
	that achieve the rate-distortion tuples $\big(R_{\text{2d},X_1|Y}(D_1,D_s), D_1,D_s\big)$ and $\big(R_{X_2|Y}(D_2), D_2\big)$, respectively. 
	Now we consider the distribution $p^*(x_1,x_2,\hat{x}_1,\hat{x}_2,\hat{s}|y)\triangleq p^*(x_1,\hat{x}_1,\hat{s}|y) p^*(x_2,\hat{x}_2|y)$ 
	which requires the Markov chain $(X_1,\hat{X}_1,\hat{S})-Y-(X_2,\hat{X}_2)$ and is consistent with the condition $X_1-Y-X_2$. 
	Then the corresponding random variables satisfy 
	\begin{align*}
	&I(X_1,X_2;\hat{X}_1,\hat{X}_2,\hat{S}|Y)  \nonumber \\
	&= H(X_1|Y)+H(X_2|Y)-H(X_1|\hat{X}_1,\hat{X}_2,\hat{S},Y)  \nonumber \\
	&-H(X_2|X_1,\hat{X}_1,\hat{X}_2,\hat{S},Y)  \\
	&= H(X_1|Y)+H(X_2|Y) \!-\! H(X_1|\hat{X}_1,\hat{S},Y) \!-\! H(X_2|\hat{X}_2,Y)  \\
	&= I(X_1;\hat{X}_1,\hat{S}|Y) + I(X_2;\hat{X}_2|Y), \\
	&= R_{\text{2d},X_1|Y}(D_1,D_s) + R_{X_2|Y}(D_2), 
	\end{align*}
	where the first equality follows from \eqref{sep_prf_indep-X2-X1Y}, the second equality follows from the Markov chain $(X_1,\hat{X}_1,\hat{S})-Y-(X_2,\hat{X}_2)$, 
	and the last equality follows from the optimality of $p^*(\hat{x}_1,\hat{s}|x_1,y)$ and $p^*(\hat{x}_2|x_2,y)$. 
	Lastly, by the minimization in the expression of the rate-distortion function in \eqref{RD-fn_rate}, we conclude that $R(D_1,D_2,D_s)\leq R_{\text{2d},X_1|Y}(D_1,D_s) + R_{X_2|Y}(D_2)$, which completes the proof of the lemma.

	\section{Proof of \Cref{lemma:RD-semantic}}\label{proof-lemma:RD-semantic}	
	As we are in the binary Hamming setting, we first calculate the values of $d_s'(x_1,\hat{s})$ (c.f. \eqref{distortion-s-alt}) by 
	\begin{align*}
	d_s'(0,0)&= \frac{1}{p(x_1=0)}\sum_{s=0,1}p(x_1=0, s)d_s(s,0) \\
	&= \frac{1}{p(x_1=0)}\big[p(x_1=0, s=0)d_s(0,0) \nonumber \\
	&\qquad + p(x_1=0, s=1)d_s(1,0)\big] \\
	&= \frac{p(x_1=0, s=1)}{p(x_1=0)}  \\
	&= p(s=1|x_1=0)  \\
	&= p.
	\end{align*}
	The other values follow similarly, and we obtain the distortion function 
	\begin{equation}
	d_s'(x_1,\hat{s})=
	\begin{cases}
	p, 	 &\text{if }\hat{s}=x_1 \\
	1-p, &\text{if }\hat{s}\neq x_1.
	\end{cases}  \label{distortion_X-S-hat}
	\end{equation}
	Then 
	\begin{align*}
	\bE d_s'(X_1,\hat{S})&= \sum_{x_1,\hat{s}}p(x_1,\hat{s})d_s'(x_1,\hat{s})  \\
	&= P(X_1\neq \hat{S}) \cdot (1-p) + P(X_1=\hat{S})\cdot p \\
	&=P(X_1\neq \hat{S})\big[1-2p\big]+p. 
	\end{align*}
	The distortion constraint $\bE d_s'(X_1,\hat{S})\leq D_s$ for $D_s\geq p$ implies $P(X_1\neq \hat{S})\leq \frac{D_s-p}{1-2p}$. 
	Now we can follow the rate-distortion evaluation for Bernoulli source and Hamming distortion in \cite{Thomas-Cover-book,Raymond-book} while only changing the probability $P(X_1\neq \hat{S})$ and obtain 
	\begin{equation}
	R_s(D_s)=\left[1-h_b\left(\frac{D_s-p}{1-2p}\right)\right]\cdot \mathds{1}_{_{p\leq D_s\leq 0.5}}.  \label{RD-fn-d_s}
	\end{equation}
	This proves the lemma. 
	

	\section{Proof of \Cref{thm:correlated}}\label{proof-thm:correlated}
	Note that separately compressing correlated sources is not optimal in general, i.e., the statement in \Cref{lemma:opt-separate-compression} does not hold here. 
	Then we need to evaluate the mutual information in \eqref{RD-fn_rate} over joint distributions, and we have 
	\begin{align}
	&I(X_1,X_2;\hat{X}_1,\hat{X}_2,\hat{S}|Y)  \nonumber \\
	&= H(X_1,X_2|Y) - H(X_1,X_2|\hat{X}_1,\hat{X}_2,\hat{S},Y)  \nonumber \\
	&= H(X_1)+H(X_2|X_1)+H(Y|X_1)-H(Y) \nonumber \\
	&- H(X_1,X_2|\hat{X}_1,\hat{X}_2,\hat{S},Y)\nonumber  \\
	&= h_b(p_1)+h_b(p_2) - H(X_1\oplus \hat{X}_1, X_2\oplus \hat{X}_2|\hat{X}_1,\hat{X}_2,\hat{S},Y)  \label{Correlated-LB_I-2} \\
	&\geq h_b(p_1)+h_b(p_2) - H(X_1\oplus \hat{X}_1, X_2\oplus \hat{X}_2) \nonumber \\
	&\geq h_b(p_1)+h_b(p_2) - H(X_1\oplus \hat{X}_1)-H(X_2\oplus \hat{X}_2) \nonumber \\
	&\geq h_b(p_1)+h_b(p_2) - h_b(D_1)-h_b(D_2), \nonumber
	\end{align}
	where $\oplus$ denote modulo 2 addition, 
	the second equality follows from the Markov chain $Y-X_1-X_2$, 
	the first inequality follows from the fact that conditioning does not increase entropy, 
	and the last inequality follows from $P(X_1\oplus \hat{X}_1=1)=P(X_1\neq \hat{X}_1)\leq D_1$, $P(X_2\oplus \hat{X}_2=1)=P(X_2\neq \hat{X}_2)\leq D_2$, and $h_b(D)$ is increasing in $D$ for $0\leq D\leq 0.5$. 
	By switching the roles of $\hat{X}_1$ and $\hat{S}$, i.e., replacing the conditional information in \eqref{Correlated-LB_I-2} by $H(X_1\oplus \hat{S}, X_2\oplus \hat{X}_2|\hat{X}_1,\hat{X}_2,\hat{S},Y)$, we can obtain similarly 
	\begin{equation*}
	I(X_1,X_2;\hat{X}_1,\hat{X}_2,\hat{S}|Y) \geq h_b(p_1)+h_b(p_2) - h_b\left(D_s^0\right)-h_b(D_2). 
	\end{equation*}
	Thus, the rate-distortion function is lower bounded by 
	\begin{equation}
	R(D_1,D_2,D_s)\geq h_b(p_1)+h_b(p_2) - h_b(\min\{D_1,D_s^0\})-h_b(D_2). \label{Correlated-lower_bound}
	\end{equation}
	We now show that the lower bound is tight by finding a joint distribution $p(x_1,x_2,y,\hat{x}_1,\hat{x}_2,\hat{s})$ that meets the distortion constraints and achieves the above lower bound. 
	\begin{figure}[!t]
		\centering
		\begin{tikzpicture}
		\node at (0.3,2) {$\hat{Z}_1\hat{Z}_2$};
		\node at (0.3,1.5) {$00$};
		\node at (0.3,0.5) {$01$};
		\node at (0.3,-0.5) {$10$};
		\node at (0.3,-1.5) {$11$};
		
		\node at (-0.5,1.5) {$q_1$};
		\node at (-0.5,0.5) {$q_2$};
		\node at (-0.5,-0.5) {$q_3$};
		\node at (-0.5,-1.5) {$q_4$};
		
		\draw (0.5,1.5)--(2.5,1.5)--(0.5,0.5)--(2.5,0.5)--(0.5,1.5)--(2.5,-0.5)--(0.5,-0.5)--(2.5,1.5)--(0.5,-1.5)--(2.5,-1.5)--(0.5,1.5)--(2.5,-0.5)--(0.5,0.5)--(2.5,-1.5)--(0.5,-0.5)--(2.5,0.5)--(0.5,-1.5)--(2.5,-0.5)--cycle;
		
		\node at (2.7,2) {$Z_1Z_2$};
		\node at (2.7,1.5) {$00$};
		\node at (2.7,0.5) {$01$};
		\node at (2.7,-0.5) {$10$};
		\node at (2.7,-1.5) {$11$};
		
		\node [right] at (3.0,1.5) {$(1-p_1)(1-p_2)$};
		\node [right] at (3.4,0.5) {$p_1(1-p_2)$};
		\node [right] at (3.6,-0.5) {$p_1p_2$};
		\node [right] at (3.2,-1.5) {$(1-p_1)p_2$};
		
		%
		\end{tikzpicture}
		\caption{Test channel from $\hat{Z}_1\hat{Z}_2$ to $Z_1Z_2$: 
			$Z_1\sim$Ber($p_1$), $Z_2\sim$Ber($p_2$), $Z_1,Z_2$ are mutually independent, $\hat{Z}_1\hat{Z}_2\sim(q_1,q_2,q_3,q_4)$, and the transition probability $p(z_1,z_2|\hat{z}_1,\hat{z}_2)$ is given by \eqref{Correlated-TestChannel_prob}.}
		\label{fig_TestChannel4to4}
	\end{figure}
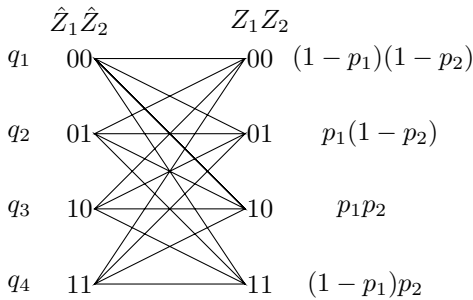
	Let $Z_i=Y\oplus X_i$ and $\hat{Z}_i=Y\oplus \hat{X}_i$ for $i=1,2$. 
	Then there is a one-to-one correspondence between $p(x_1,x_2,y,\hat{x}_1,\hat{x}_2,\hat{s})$ and $p(z_1,z_2,y,\hat{z}_1,\hat{z}_2,\hat{s})$. 
	From $Z_1,Z_2$ generated from the source distribution $p(x_1,x_2,y)$, it is easy to check that $(Z_1,Z_2)\sim[(1-p_1)(1-p_2),p_1(1-p_2),p_1p_2,(1-p_1)p_2]$ as shown in Fig.~\ref{fig_TestChannel4to4}. 
	
	Next, we construct the desired joint distribution using the test channel in Fig.~\ref{fig_TestChannel4to4} as follows. 
	For $p_1,p_2\in[0,0.5]$, $(D_1,D_2,D_s)\in\cD_0$ (c.f. \eqref{Correlated-distortion_region}), and $D_1\leq D_s^0$, 
	consider the joint distribution $p(x_1,x_2,y,\hat{x}_1,\hat{x}_2,\hat{s})$ defined by the following conditions 
	\begin{enumerate}[i.]
		\item $\hat{S}=\hat{X}_1$; 
		
		\item Markov chain $Y-(\hat{Z}_1\hat{Z}_2)-(Z_1Z_2)$; 
		
		\item The test channel in Fig.~\ref{fig_TestChannel4to4}  with the conditional probability $p(z_1,z_2|\hat{z}_1,\hat{z}_2)$ given as 
		\begin{align}
		&p(z_1,z_2|\hat{z}_1,\hat{z}_2)=  \nonumber \\
		&\begin{cases}
		(1-D_1)(1-D_2), &\text{ if } z_1=\hat{z}_1 \text{ and } z_2=\hat{z}_2  \\ 
		(1-D_1)D_2, &\text{ if } z_1=\hat{z}_1 \text{ and } z_2\neq \hat{z}_2  \\ 
		D_1(1-D_2), &\text{ if } z_1\neq \hat{z}_1 \text{ and } z_2=\hat{z}_2  \\ 
		D_1D_2, &\text{ if } z_1\neq \hat{z}_1 \text{ and } z_2\neq \hat{z}_2;  \\ 
		\end{cases}  \label{Correlated-TestChannel_prob}
		\end{align}
		
		\item In order for $Z_1,Z_2$ to follow the independent Bernoulli distributions, we need to choose the distribution of $\hat{Z}_1\hat{Z}_2$ as \eqref{distribution-Zs-1}-\eqref{distribution-Zs-4}. 
		\begin{figure*}
			\begin{align}
			q_1 &= \frac{[(1-p_2)-D_1][(1-p_1-p_2+2p_1p_2)-D_2]+p_2(1-2p_1)(1-p_2)}{(1-2D_1)(1-2D_2)}  \label{distribution-Zs-1}\\
			q_2 &= \frac{[(1-p_2)-D_1][(p_1+p_2-2p_1p_2)-D_2]+p_2(2p_1-1)(1-p_2)}{(1-2D_1)(1-2D_2)} \\
			q_3 &= \frac{(p_2-D_1)[(1-p_1-p_2+2p_1p_2)-D_2]+p_2(2p_1-1)(1-p_2)}{(1-2D_1)(1-2D_2)} \\
			q_4 &= \frac{(p_2-D_1)[(p_1+p_2-2p_1p_2)-D_2]+p_2(1-2p_1)(1-p_2)}{(1-2D_1)(1-2D_2)} \label{distribution-Zs-4}
			\end{align}
		\end{figure*}
		We can verify that $q_i\geq 0,~i=1,2,3,4$ for $(D_1,D_2,D_s)\in\cD_0$. 
	\end{enumerate}
	Now it remains to verify that the above distribution achieves the rate in \eqref{Correlated-lower_bound} and distortions $D_1,D_2$, and $D_s$. 
	From conditions i and ii, we have 
	\begin{align*}
	&I(X_1,X_2;\hat{X}_1,\hat{X}_2,\hat{S}|Y)  \nonumber \\
	&= H(X_1,X_2|Y) - H(X_1,X_2|\hat{X}_1,\hat{X}_2,Y)  \\
	&= H(X_1,X_2|Y) - H(X_1,X_2|\hat{X}_1,\hat{X}_2)  \\
	&= h_b(p_1)+h_b(p_2) - h_b(D_1)-h_b(D_2). 
	\end{align*}
	From \eqref{Correlated-TestChannel_prob} and \eqref{distortion_X-S-hat}, it is easy to calculate the expected distortions as  
	\begin{align}
	\bE d_1(X_1,\hat{X}_1)&= D_1 \\
	\bE d_2(X_2,\hat{X}_2)&= D_2 \\
	\bE d_s'(X_1,\hat{S})&=(1-D_1)p+D_1(1-p) \leq D_s. 
	\end{align}
	
	On the other hand, if $D_1\geq D_s^0$, we can construct the joint distribution $p(x_1,x_2,y,\hat{x}_1,\hat{x}_2,\hat{s})$ using four conditions similarly 
	by switching the role of $(\hat{X}_1,D_1)$ and $(\hat{S},D_s^0)$. 
	Then the rate and distortions can be obtained accordingly. 
	This proves the theorem. 

	\section{Proof of \Cref{thm:classify}}\label{proof-thm:classify}
	Since \Cref{lemma:opt-separate-compression} holds here, we first calculate the rate-distortion function for $X_2$, which is 
	\begin{align}
	R_{X_2|Y}(D_2) &= \big[1-h_b(D_2)\big]\cdot \mathds{1}_{_{0\leq D_2\leq 0.5}}. 
	\end{align}
	Now it remains to calculate $R_{\text{2d},X_1|Y}(D_1,D_s)$. 
	We first consider the case that $D_1\leq D_s^0$ and provide a lower bound of the mutual information as follows
	\begin{align*}
	&I(X_1;\hat{X_1},\hat{S}|Y)  \nonumber \\
	&\geq I(X_1;\hat{X_1}|Y)   \\
	&= H(X_1|Y) - H(X_1|\hat{X}_1,Y)  \\
	&\geq H(X_1|Y) - H(X_1|\hat{X}_1)  \\
	&= h_b(p_2)+\log(N/2) - H(X_1|\hat{X}_1)  \\
	&\geq h_b(p_2)+\log(N/2) - h_b(D_1)-D_1\log(N-1), 
	\end{align*}
	where the last inequality follows from $P(\hat{X}_1\neq X_1)\leq D_1$, $h_b(x)$ is an increasing function for $x\in[0,0.5]$, and the fact that uniform distribution maximizes entropy. 
	(Note that we can also obtain the above lower bound by directly applying the log-sum inequality.) 
	
	Next, we show the lower bound is tight by finding a joint distribution that achieves the above rate and distortions $D_1$ and $D_s$. 
	For $0\leq D_1\leq \frac{2(N-1)p_2}{N}$, we choose $\hat{S}=\hat{X}_1$ and $(X_1,Y,\hat{X}_1)$ by the test channel $p(x_1|\hat{x}_1)$ in Fig.~\ref{fig_TestChannelNtoN} and the Markov chain $Y-\hat{X}_1-X_1$. 
	The conditional probability of the test channel in Fig.~\ref{fig_TestChannelNtoN} is given as 
	\begin{align}
	p(\hat{x}_1|x)=
	\begin{cases}
	1-D_1, &\text{ if }\hat{x}_1=x  \\
	\frac{D_1}{N-1}, &\text{ if }\hat{x}_1\neq x. 
	\end{cases} \label{Classify-TestChannel_prob}
	\end{align}
	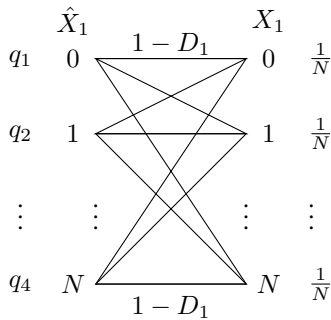
\begin{figure}[!t]
		\centering
		\begin{tikzpicture}
		\node at (1.5,1.7) {$1-D_1$};
		\node at (1.5,-1.8) {$1-D_1$};
		\node at (0.2,2) {$\hat{X}_1$};
		\node at (0.2,1.5) {$0$};
		\node at (0.2,0.5) {$1$};
		\node at (0.5,-0.5) {$\vdots$};
		\node at (0.2,-1.5) {$N$};
		
		\node at (-0.5,1.5) {$q_1$};
		\node at (-0.5,0.5) {$q_2$};
		\node at (-0.5,-0.5) {$\vdots$};
		\node at (-0.5,-1.5) {$q_4$};
		
		\draw (0.5,1.5)--(2.5,1.5)--(0.5,0.5)--(2.5,0.5)--cycle; 
		\draw (0.5,1.5)--(2.5,-1.5)--(0.5,-1.5)--(2.5,1.5)--cycle; 
		\draw (0.5,0.5)--(2.5,0.5)--(0.5,-1.5)--(2.5,-1.5)--cycle;
		
		\node at (2.8,2) {$X_1$};
		\node at (2.8,1.5) {$0$};
		\node at (2.8,0.5) {$1$};
		\node at (2.5,-0.5) {$\vdots$};
		\node at (2.8,-1.5) {$N$};
		
		\node [right] at (3.2,1.5) {$\frac{1}{N}$};
		\node [right] at (3.2,0.5) {$\frac{1}{N}$};
		\node [right] at (3.2,-0.5) {$\vdots$};
		\node [right] at (3.2,-1.5) {$\frac{1}{N}$};
		
		\end{tikzpicture}
		\caption{Test channel from $\hat{X}_1$ to $X_1$: 
			$X_1\sim$Uniform($\frac{1}{N}$) and the transition probability $p(x_1|\hat{x}_1)$ is given by \eqref{Classify-TestChannel_prob}.}
		\label{fig_TestChannelNtoN}
	\end{figure}
	%
	Solving the equations 
	\begin{align*}
	q_i(1-D_1)+\left(\sum_{j\in[1:N],j\neq i}q_j\right)\frac{D_1}{N-1} = \frac{1}{N}, ~i\in[1:N], 
	\end{align*}
	we obtain that $q_i=\frac{1}{N}$ for $i\in[1:N]$, i.e., $\hat{X}_1$ is also uniformly distributed over $[1:N]$. 
	For the joint distribution of $Y$ and $\hat{X}_1$, we define 
	\begin{align*}
	p(y|\hat{x}_1)=
	\begin{cases}
	\dfrac{p_2-\frac{ND_1}{2(N-1)}}{1-\frac{ND_1}{N-1}}, &\text{if }y=0,\hat{x}_1\text{ is odd}  \\ 
	~ & \text{~~or }y=1,\hat{x}_1\text{ is even}  \\ 
	~  \\
	\dfrac{1-p_2-\frac{ND_1}{2(N-1)}}{1-\frac{ND_1}{N-1}}, &\text{if }y=0,\hat{x}_1\text{ is even} \\
	~ & \text{~~or }y=1,\hat{x}_1\text{ is odd}.
	\end{cases}
	\end{align*}
	We see that $p(y|\hat{x}_1)\geq 0$ for any $D_1\leq \frac{2(N-1)p_2}{N}$, i.e., $(D_1,D_2,D_s)\in\cD_1$. 
	Then we can verify using $p(y|x_1)=\sum_{\hat{x}_1}p(y|x_1,\hat{x}_1)p(\hat{x}|x)=\sum_{\hat{x}_1}p(y|\hat{x}_1)p(\hat{x}|x)$ that the above distribution 
	can induce the same conditional probability $p(y|x_1)=p_2,1-p_2$ as defined at the beginning of this section. 
	Thus, we have constructed a feasible $p(x_1,y,\hat{x}_1)$ that can achieve expected distortion $\bE d_1(X,\hat{X}_1) = D_1$, $D_s^0\geq D_1$, and mutual information 
	\begin{align*}
	&I(X_1;\hat{X_1},\hat{S}|Y)  \nonumber \\
	&= I(X_1;\hat{X_1}|Y)   \\
	&= H(X_1|Y) - H(X_1|\hat{X}_1,Y)  \\
	&= H(X_1|Y) - H(X_1|\hat{X}_1)  \\
	&= h_b(p_2)+\log(N/2) - h_b(D_1)-D_1\log(N-1), 
	\end{align*}
	where the first equality follows from $\hat{S}=\hat{X}_1$, 
	the third equality follows from the Markov chain $Y-\hat{X}_1-X_1$, 
	and the last equality follows from the joint distribution of $(X_1,Y)$ and the distribution in \eqref{Classify-TestChannel_prob}. 
	
	For the other case that $D_1\geq D_s^0$, we only need to switch the role of $(\hat{X}_1,D_1)$ and $(\hat{S},D_s^0)$. 
	Then the rate and distortions can be obtained accordingly, which can prove the theorem.

	\section{Proof of \Cref{thm:Gaussian}}\label{proof-thm:Gaussian}
	The rate-distortion function in Theorem 1 satisfies
	\begin{equation}
	R(D_1,D_2,D_s)= R_{2d,X_1|Y}(D_1,D_s)+R_{X_2|Y}(D_2).
	\end{equation}
	The second term is the solution to the quadratic Gaussian source coding problem with side information \cite[Chapter 11]{NetworkIT-book}, given as 
	\begin{equation}\label{eqn:X2 rate}
	R_{X_2|Y}(D_2)=\frac{1}{2}\left(\log\frac{\sigma_{X_2|Y}}{D_2}\right)^+,
	\end{equation}
	where $\sigma_{X_2|Y}$ is the conditional variance of $X_2$ given $Y$. 
	Note that $X_2$ is Gaussian conditioning on $Y$, i.e.,  
	$X_2|Y\sim\mathcal{N}(\frac{\sigma_{X_2Y}}{\sigma_{Y}}Y,\sigma_{X_2}-\frac{\sigma_{X_2Y}^2}{\sigma_{Y}})$, 
	which implies $\sigma_{X_2|Y}=\sigma_{X_2}-\frac{\sigma_{X_2Y}^2}{\sigma_{Y}}$.
	
	For the first term, note that $R_{2d,X_1|Y}(D_1,D_s)$ is lowered bounded by both 
	\begin{equation}\label{eqn:X1 rate}
	R_{X_1|Y}(D_1)=\min_{\mathbb{E}d_1\left(X_1,\hat{X}_1\right)\le D_1} I(X_1;\hat{X}_1|Y),
	\end{equation}
	and
	\begin{equation}\label{eqn:indirect rate}
	R_{S|Y}(D_s)=\min_{\mathbb{E}d_S^\prime\left(X_1,\hat{S}\right)\le D_s} I(X_1;\hat{S}|Y).
	\end{equation}
	Obviously, (\ref{eqn:X1 rate}) is the solution of the quadratic Gaussian source coding with side information similarly to (\ref{eqn:X2 rate}).
	
	We proceed to calculate (\ref{eqn:indirect rate}), which is actually the semantic rate-distortion function of the indirect source coding with side information. Observing 
	$\left(X_1,Y\right)$, $S$ is conditionally Gaussian as $S|\left(X_1,Y\right)\sim \mathcal{N}\left(\frac{\sigma_{SX_1}}{\sigma_{X_1}}X_1,\sigma_S-\frac{\sigma_{SX_1}^2}{\sigma_{X_1}}\right)$.
	It is shown in~\cite{Wolf-Ziv-IT70-Transmission} that we can rewrite the semantic distortion as 
	\begin{equation}\label{eqn:EC distortion}
	\mathbb{E}d_S^\prime\left(X_1,\hat{S}\right)=\mathbb{E}\left[\left(S-\tilde{S}_\mathrm{MMSE}\right)^2\right]+\mathbb{E}\left[\left(\tilde{S}_\mathrm{MMSE}-\hat{S}\right)^2\right],
	\end{equation}
	where $\tilde{S}_\mathrm{MMSE}=\frac{\sigma_{SX_1}}{\sigma_{X_1}}X_1$ is the MMSE estimator upon observing $X_1$ and $Y$, 
	and the first term on the right-hand side is the corresponding minimum mean squared error ($\mathrm{mmse}$ c.f. \eqref{eqn:mmse_def}), i.e., 
	\begin{equation}
	\mathbb{E}\left[\left(S-\tilde{S}_\mathrm{MMSE}\right)^2\right] =\mathrm{mmse} =\sigma_S-\frac{\sigma_{SX_1}^2}{\sigma_{X_1}}.
	\end{equation}
	Then we consider a specific encoder which first estimates the semantic information using MMSE estimator 
	and then compresses the estimation under mean squared error distortion constraint $D_s-\mathrm{mmse}$ with side information. 
	The resulting achievable rate provides an upper bound on $R_{S|Y}(D_s)$ which is
	\begin{align}\label{eqn:UB}
	R_{S|Y}(D_s)&\le \frac{1}{2}\left(\log\frac{\sigma_{\tilde{S}_\mathrm{MMSE}|Y}}{D_s-\mathrm{mmse}}\right)^+  \nonumber \\
	&=\frac{1}{2}\left(\log\frac{\sigma_{SX_1}^2\sigma_{X_1|Y}}{\sigma_{X_1}^2\left(D_s-\mathrm{mmse}\right)}\right)^+  \nonumber \\
	&=\frac{1}{2}\left(\log\frac{\sigma_{SX_1}^2\left(\sigma_{X_1}-\frac{\sigma_{X_1Y}^2}{\sigma_{Y}}\right)}{\sigma_{X_1}^2\left(D_s-\mathrm{mmse}\right)}\right)^+.
	\end{align}
	Furthermore, for $D_s\ge \mathrm{mmse}+\frac{\sigma_{SX_1}^2\sigma_{X_1|Y}}{\sigma_{X_1}^2}$, 
	we have $R_{S|Y}(D_s)=0$, which can be obtained by setting 
	$\hat{S}=\mathbb{E}\left[\tilde{S}_\mathrm{MMSE}|Y\right]=\frac{\sigma_{SX_1}}{\sigma_{X_1}}\mathbb{E}\left[X_1|Y\right]$. 
	For $D_s< \mathrm{mmse}+\frac{\sigma_{SX_1}^2\sigma_{X_1|Y}}{\sigma_{X_1}^2}$, we derive a lower bound for $R_{S|Y}(D_s)$ as follows 
	\begin{align}
	&R_{S|Y}(D_s)\ge I(X_1;\hat{S}|Y)=H(X_1|Y)-H(X_1|\hat{S},Y)  \nonumber \\
	&=\frac{1}{2}\log\left(2\pi e \sigma_{X_1|Y}\right)-H(X_1-\frac{\sigma_{X_1}}{\sigma_{SX_1}}\hat{S}|\hat{S},Y) \nonumber \\
	&\ge \frac{1}{2}\log\left(2\pi e \sigma_{X_1|Y}\right)-H(X_1-\frac{\sigma_{X_1}}{\sigma_{SX_1}}\hat{S})  \nonumber \\
	&\ge \frac{1}{2}\log\left(2\pi e \sigma_{X_1|Y}\right)-\frac{1}{2}\log\left(2\pi e \mathbb{E}\left[\left(X_1-\frac{\sigma_{X_1}}{\sigma_{SX_1}}\hat{S}\right)^2\right]\right)   \label{eqn:LB-4}\\
	&\ge \frac{1}{2}\log\left(2\pi e \sigma_{X_1|Y}\right)-\frac{1}{2}\log\left(2\pi e \frac{\sigma_{X_1}^2\left(D_s-\mathrm{mmse}\right)}{\sigma_{SX_1}^2}\right)  \label{eqn:LB-5} \\
	&= \frac{1}{2}\log\frac{\sigma_{SX_1}^2\sigma_{X_1|Y}}{\sigma_{X_1}^2\left(D_s-\mathrm{mmse}\right)}  \nonumber \\
	&= \frac{1}{2}\log\frac{\sigma_{SX_1}^2\left(\sigma_{X_1}-\frac{\sigma_{X_1Y}^2}{\sigma_{Y}}\right)}{\sigma_{X_1}^2\left(D_s-\mathrm{mmse}\right)}  \label{eqn:LB}
	\end{align}
	where \eqref{eqn:LB-4} is due to the fact that the Gaussian distribution maximizes the entropy for a given variance, 
	and \eqref{eqn:LB-5} follows from (\ref{eqn:EC distortion}) and the semantic distortion constraint. 
	Combining the upper and lower bounds in \eqref{eqn:UB} and \eqref{eqn:LB}, we obtain 
	\begin{equation}
	R_{S|Y}(D_s)=\frac{1}{2}\left(\log\frac{\sigma_{SX_1}^2\left(\sigma_{X_1}-\frac{\sigma_{X_1Y}^2}{\sigma_{Y}}\right)}{\sigma_{X_1}^2\left(D_s-\mathrm{mmse}\right)}\right)^+.
	\end{equation}
	Thus we have 
	\begin{align}
	\begin{split}
	&R(D_1,D_2,D_s)\ge  \max\Big\{R_{X_1|Y}(D_1),R_{S|Y}(D_s)\Big\}+R_{X_2|Y}(D_2) \\
	&=\frac{1}{2}\left(\log\frac{\sigma_{X_2}-\frac{\sigma_{X_2Y}^2}{\sigma_{Y}}}{D_2}\right)^+ + \nonumber  \\ 
	&\frac{1}{2}\left[\log\max\left(\frac{\sigma_{X_1}-\frac{\sigma_{X_1Y}^2}{\sigma_{Y}}}{D_1}, \frac{\sigma_{SX_1}^2\left(\sigma_{X_1}-\frac{\sigma_{X_1Y}^2}{\sigma_{Y}}\right)}{\sigma_{X_1}^2\left(D_s-\mathrm{mmse}\right)}\right)\right]^+.
	\end{split}
	\end{align}
	
	To show the achievability, consider the following two cases. 
	\begin{itemize}
		\item For $\frac{D_s-\mathrm{mmse}}{\sigma_{SX_1}^2}\ge\frac{D_1}{\sigma_{X_1}^2}$, 
		we first reconstruct $\hat{X}_1$ and $X_2$ subject to distortion constraints $D_1$ and $D_2$, 
		and hence achieve $R_{X_1|Y}(D_1)+R_{X_2|Y}(D_2)$. 
		Then we recover the semantic information by $\hat{S}=\frac{\sigma_{SX_1}}{\sigma_{X_1}}\hat{X}_1$, and the semantic distortion satisfies
		\begin{align*}
		\mathbb{E}d_S^\prime\left(X_1,\hat{S}\right)&=\mathrm{mmse}+\mathbb{E}\left[\left(\frac{\sigma_{SX_1}}{\sigma_{X_1}}X_1-\frac{\sigma_{SX_1}}{\sigma_{X_1}}\hat{X}_1\right)^2\right] \\
		&\le \mathrm{mmse}+\frac{\sigma_{SX_1}^2}{\sigma_{X_1}^2}D_1 \\
		&\le D_s.
		\end{align*}
		
		\item For $\frac{D_s-\mathrm{mmse}}{\sigma_{SX_1}^2}<\frac{D_1}{\sigma_{X_1}^2}$, 
		we first reconstruct $\hat{S}$ and $X_2$ subject to distortion constraints $D_s$ and $D_2$, 
		and hence achieve $R_{S|Y}(D_s)+R_{X_2|Y}(D_2)$. 
		Then we recover $\hat{X}_1=\frac{\sigma_{X_1}}{\sigma_{SX_1}}\hat{S}$, and the distortion satisfies
		\begin{align*}
		\mathbb{E}d_1\left(X_1,\hat{X}_1\right)&=\mathbb{E}\left[\left(X_1-\frac{\sigma_{X_1}}{\sigma_{SX_1}}\hat{S}\right)^2\right] \\
		&=\frac{\sigma_{X_1}^2}{\sigma_{SX_1}^2}\mathbb{E}\left[\left(\tilde{S}_\mathrm{MMSE}-\hat{S}\right)^2\right] \\
		&=\frac{\sigma_{X_1}^2}{\sigma_{SX_1}^2}\left(\mathbb{E}d_S^\prime\left(X_1,\hat{S}\right)-\mathrm{mmse}\right) \\
		&< D_1.
		\end{align*}
	\end{itemize}
	This establishes the achievability and thus completes the proof.

	\bibliographystyle{IEEEtran}
	\bibliography{ref_segmentation}
	
\end{document}